\documentclass[12pt]{iopart}
\usepackage{cite}
\usepackage{graphicx,color}
\usepackage{iopams}
\usepackage{amsthm}
\usepackage[normalem]{ulem}
\usepackage{hyperref}
\usepackage{bm}

\eqnobysec

\newtheorem{lemma}{Lemma}

\renewcommand{\Re}{\mathop{\mathrm{Re}}\nolimits}

\newcommand{\ket}[1]{|{#1}\rangle}
\newcommand{\bra}[1]{\langle{#1}|}
\newcommand{\bras}[2]{{}_{#2}\hspace*{-0.2mm}\langle{#1}|}

\newcommand{\bracket}[2]{\langle#1|#2\rangle}

\newcommand{\openone}{\mathbb{I}}

\definecolor{dgreen}{rgb}{0,0.5,0}

\definecolor{delete}{cmyk}{0.5,0,0,0}

\begin{document}
\title[Optimal Gaussian Metrology for Generic Multimode Interferometric Circuit]{Optimal Gaussian Metrology for Generic Multimode Interferometric Circuit} 

\author{Teruo Matsubara${}^1$, Paolo Facchi${}^{2,3}$, Vittorio Giovannetti${}^4$ and Kazuya Yuasa${}^1$}
\address{${}^1$ Department of Physics, Waseda University, Tokyo 169-8555, Japan}
\address{${}^2$ Dipartimento di Fisica and MECENAS, Universit\`a di Bari, I-70126 Bari, Italy}
\address{${}^3$ INFN, Sezione di Bari, I-70126 Bari, Italy}
\address{${}^4$ NEST, Scuola Normale Superiore and Istituto Nanoscienze-CNR, I-56127 Pisa, Italy}

\date\today

\begin{abstract}
Bounds on the ultimate precision attainable in the estimation of a parameter in Gaussian quantum metrology are obtained when the average number of bosonic probes is fixed.
We identify the optimal input probe state among generic (mixed in general) Gaussian states with a fixed average number of probe photons for the estimation of a parameter contained in a generic multimode interferometric optical circuit, namely, a passive linear circuit preserving the total number of photons. 
The optimal Gaussian input state is essentially a single-mode squeezed vacuum, and the ultimate precision is achieved by a homodyne measurement on the single mode.
We also reveal the best strategy for the estimation when we are given $L$ identical target circuits and are allowed to apply passive linear controls in between with an arbitrary number of ancilla modes introduced.
\end{abstract}

\maketitle

\section{Introduction}
Quantum-mechanical features and quantum effects can drastically improve the accuracy of measurements~\cite{ref:MetrologyScience,ref:QuantumMetrologyVittorio,doi:10.1080/00107510802091298,ref:MetrologyNaturePhoto,DemkowiczDobrzanski2015345,JLightwaveTechno.33.2359}.
This is known as \textit{quantum metrology}, and is one of the promising future quantum technologies.
In particular, quantum optical measurement schemes using photonic probes have recently been under intense study~\cite{ref:MetrologyScience,doi:10.1080/00107510802091298,ref:MetrologyNaturePhoto,DemkowiczDobrzanski2015345,JLightwaveTechno.33.2359}, pursuing strategies that allow  to beat the standard quantum limit on the measurement accuracy, both theoretically~\cite{PhysRevD.23.1693,PhysRevD.30.2548,PhysRevA.33.4033,PhysRevLett.71.1355,PhysRevLett.75.2944,PhysRevLett.85.5098,PhysRevA.73.011801,PhysRevA.73.033821,PhysRevA.76.013804,PhysRevLett.100.073601,PhysRevLett.101.253601,PhysRevLett.102.040403,PhysRevA.79.033834,PhysRevLett.102.253601,PhysRevLett.104.103602,PhysRevLett.105.120501,Escher:2011aa,PhysRevLett.107.083601,PhysRevA.85.010101,PhysRevA.85.011801,ref:RivasLuis-NJP2012,PhysRevA.87.012107,arXiv:1303.3682,PhysRevLett.110.163604,PhysRevLett.110.213601,PhysRevA.88.013838,PhysRevLett.111.173601,PhysRevA.88.040102,PhysRevA.88.063820,PhysRevA.89.032128,PhysRevA.89.053822,PhysRevA.90.025802,PhysRevA.90.033846,PhysRevA.91.013808,PhysRevA.91.032103,PhysRevLett.114.170802,Sparaciari:15,PhysRevA.92.042331,1367-2630-17-7-073016,PhysRevA.92.022106,ref:EstId,PhysRevA.93.013809,PhysRevA.93.023810,PhysRevA.93.033859,PhysRevA.94.023834,PhysRevX.6.031033,PhysRevA.94.033817,PhysRevA.94.042327,PhysRevA.94.042342,PhysRevLett.117.190801,PhysRevLett.117.190802,PhysRevX.6.041044,PhysRevA.94.062313,PhysRevA.95.012109,ref:MultiParGaussianMetrology,arXiv:1701.05152,arXiv:1602.05958,arXiv:1801.00299} and experimentally~\cite{Bouwmeester:2004aa,Nagata726,Higgins:2007aa,OBrien:2009aa,BridaG.:2010aa,Afek879,KacprowiczM.:2010aa,Xian:2011aa,PhysRevLett.107.080504,PhysRevA.85.043817,doi:10.1063/1.4724105,Wolfgramm:2013aa,Taylor:2013aa,Ono:2013aa,Vidrighin:2014aa,PhysRevLett.112.103604,PhysRevLett.112.223602}.

In a variety of quantum optical metrology settings, the probe sensitivity to the target parameter can be improved by squeezing the state of the input light~\cite{PhysRevD.23.1693,PhysRevD.30.2548}.
Entanglement is also an important keyword in the studies of quantum metrology~\cite{ref:MetrologyNaturePhoto,DemkowiczDobrzanski2015345,JLightwaveTechno.33.2359}.
In these ways, the state of the input probe photons is important for high precision metrology.

There is an interesting class of states of light: \textit{Gaussian states}. From a practical point of view, a variety of Gaussian states are relatively easy to generate in laboratories, and various quantum information tasks have been implemented experimentally using photons in Gaussian states~\cite{ref:ContVarQI,ref:GaussianQI,ref:ContVarQI-Adesso}.
Also from a theoretical point of view, they provide an interesting category of quantum information protocols~\cite{ref:ContVarQI,ref:GaussianQI,ref:ContVarQI-Adesso}. For these reasons, quantum optical metrology with Gaussian input probe states and/or Gaussian channels has been eagerly investigated~\cite{PhysRevA.94.062313,PhysRevA.88.040102,PhysRevA.94.042342,PhysRevA.95.012109,PhysRevA.79.033834,PhysRevA.94.023834,PhysRevA.87.012107,PhysRevA.88.013838,PhysRevA.89.032128,1367-2630-17-7-073016,PhysRevA.92.022106,PhysRevLett.100.073601,PhysRevLett.104.103602,PhysRevLett.101.253601,PhysRevA.85.010101,PhysRevA.73.033821,Sparaciari:15,PhysRevA.93.023810,ref:EstId,arXiv:1303.3682,ref:MultiParGaussianMetrology,arXiv:1801.00299}.

For instance, the estimation of a single-mode phase shift is studied with pure~\cite{PhysRevA.73.033821} and mixed~\cite{PhysRevA.79.033834} Gaussian probes, and some other single-mode Gaussian channels such as squeezing and amplitude-damping are analyzed with general mixed Gaussian probes~\cite{PhysRevA.88.040102}. The estimation of a single-mode phase shift with general mixed Gaussian probes is discussed in the presence of general Gaussian dissipation~\cite{PhysRevA.95.012109}. A few specific two-mode Gaussian channels like two-mode squeezing and mode-mixing are studied with some particular types of two-mode Gaussian probes~\cite{PhysRevA.94.062313}. The ultimate precision bound is clarified for generic two-mode passive linear circuits, which preserve the number of photons passing through them (they are Gaussian channels)~\cite{ref:EstId}. A formula for the quantum Fisher information valid for any multimode pure Gaussian states is derived and investigated under the condition of intense probe light (with large displacement)~\cite{PhysRevA.85.010101}. 
General multimode Gaussian unitary channels (Bogoliubov transformations) are considered with pure probe states not restricted to Gaussian states and the behavior of the quantum Fisher information for large mean photon numbers is discussed~\cite{PhysRevA.92.022106}.
A formula for the quantum Fisher information matrix is derived for general multimode Gaussian states and multiparameter Gaussian quantum metrology is discussed~\cite{ref:MultiParGaussianMetrology,arXiv:1801.00299}.

In this paper, we study the estimation of a parameter embedded in a generic $M$-mode passive linear interferometric circuit, and clarify the ultimate precision bound achievable with Gaussian probes.
We identify the optimal input probe state among all Gaussian states (including mixed Gaussian states) with a fixed average number of probe photons.
Such a bound is known for $M=2$~\cite{ref:EstId}, but is not known for $M\ge3$.
The proof strategy taken for $M=2$ is not helpful for $M\ge3$, and it is not a simple generalization of the previous work.

\begin{figure}[t]
\centering
\includegraphics[scale=1.2]{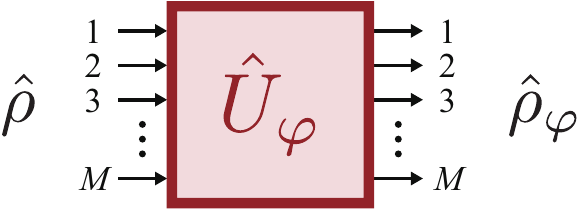}
\caption{The generic passive linear optical circuit $\hat{U}_\varphi$ with $M$ input ports and $M$ output ports. Our problem is to estimate a parameter $\varphi$ contained in the circuit $\hat{U}_\varphi$, by sending probe photons through it and measuring its output. We will restrict ourselves to Gaussian input states $\hat{\rho}$ with a given average number of probe photons $\langle\hat{N}\rangle=\overline{N}$, among which we identify the best Gaussian states reducing the accuracy limit in the estimation of $\varphi$ as much as possible.}
\label{fig:BlackBox}
\end{figure}
More specifically, we will consider the setting shown in Fig.~\ref{fig:BlackBox}: a collection of $M$ photonic modes is employed as a probe to recover the value of an unknown parameter $\varphi$, which is imprinted on the state of the probe via the action of a  passive (i.e.~photon-number preserving), Gaussian (i.e.~mapping
Gaussian input into Gaussian output), unitary transformation 
$\hat{U}_\varphi$.
 Under the assumption that the allowed  input density matrices $\hat{\rho}$  of the $M$ modes belong to the set $\mathcal{G}(M,\overline{N})$  of 
(not necessarily pure) Gaussian states  with an average photon number  $\overline{N}$, we are interested in the 
ultimate accuracy  in the estimation of $\varphi$ attainable when having full access to the output state 
\begin{equation}
\hat{\rho}_\varphi = \hat{U}_\varphi \hat{\rho} \hat{U}_\varphi^\dag.\label{OUTPUT}
\end{equation} 
Our main result consists in showing that, irrespective of the explicit form of $\hat{U}_\varphi$, 
the minimum value  of the  uncertainty  $\delta \varphi$  on the estimation of $\varphi$ is bounded from below by the Heisenberg-like scaling 
\begin{equation} \label{HEIS}
\delta \varphi_{\mathrm{min}} \gtrsim 1/\overline{N}. 
\end{equation} 
To this end,
we shall focus on the quantum Fisher information (QFI)  $F(\varphi|\hat{\rho})$ of the problem, 
 which, via  the quantum  Cram\'er-Rao inequality~\cite{ref:Helstrom,ref:BraunsteinCave1994,ref:BraunsteinCave1996AnnPhys,ref:MetrologyScience,ref:HayashiAsymptoticTheory,ref:Paris-IJQI,ref:MetrologyNaturePhoto,ref:HolevoSNS}, 
 sets a universal bound on $\delta\varphi_{\mathrm{min}}$ that is 
 independent of the adopted measurement procedure,
 \begin{equation} \label{CRQQQQ} 
\delta \varphi_{\mathrm{min}} \ge\frac{1}{\sqrt{F(\varphi|\hat{\rho})}}.
\end{equation}
 We hence prove~(\ref{HEIS}) by showing  that  the maximum value of $F(\varphi|\hat{\rho})$ attainable on the set $\mathcal{G}(M,\overline{N})$ 
 is bounded by a quantity which scales quadratically in $\overline{N}$, namely,
\begin{equation}
{F}(\varphi|\hat{\rho})\leq 
8 \|g_\varphi\|^2\overline{N}(\overline{N}+1), \quad \forall \hat{\rho} \in \mathcal{G}(M,\overline{N}).
\label{eqn:OptQFISeq}
\end{equation}
Here,  $\|g_\varphi\|$ is the spectral norm of the Hermitian matrix
\begin{equation}
g_\varphi
=iU_\varphi^\dag\frac{\rmd U_\varphi}{\rmd\varphi},
\label{eqn:Hamiltonian0}
\end{equation}
with $U_\varphi$ being  the unitary matrix describing the circuit, defined in~(\ref{eqn:UnitaryRotation}), and is independent of the input state $\hat{\rho}$.

Moreover, we show that the bound~(\ref{eqn:OptQFISeq}) is sharp and can be saturated. In fact, we  identify 
 the optimal states within $\mathcal{G}(M,\overline{N})$ that saturate the inequality~(\ref{eqn:OptQFISeq}): they are pure states~$\rho_\mathrm{opt}=\ket{\psi_\mathrm{opt}}\bra{\psi_\mathrm{opt}}$ given in~(\ref{eqn:OptPure}). 
We note that, apart from some special cases,
such optimal vectors $\ket{\psi_\mathrm{opt}}$ generally depend on the variable $\varphi$, whose unknown value we wish to determine. Therefore, the possibility of using this optimal input state for achieving the bound is not straightforward, and would require in practice the use of iterative procedures with a sequence of input states that approximate the optimal state.
Anyway, the optimal state $\ket{\psi_\mathrm{opt}}$ enables us to reach the upper bound~(\ref{eqn:OptQFISeq}).

The paper is organized as follows.
The model and the estimation problem are set up in Sec.~\ref{sec:model}\@.
In Sec.~\ref{sec:Optimization}, the maximal precision achievable by a Gaussian probe is found, first for pure Gaussian states and then for mixed Gaussian states. Moreover, we explicitly find the optimal states that achieve the maximal precision.
Two different measurement schemes are presented in Sec.~\ref{sec:Measurements}\@.
We look at a few simple examples in Sec.~\ref{sec:Examples}\@.
Furthermore, in Sec.~\ref{sec:Sequential}, we exhibit the optimal sequential strategy for the estimation when several target circuits, together with ancilla modes, are allowed to be used.
A summary of the present work is given in Sec.~\ref{sec:Conclusion}\@.
We add four appendices, containing some technical tools and proofs. 
In \ref{sec:Setup} we collect some results on Gaussian states and operations, in \ref{app:3.17} we show the derivation of a formula for the QFI, in \ref{app:Ineq} we prove some inequalities on Hermitian matrices used in the solution of the optimization problems, and \ref{app:OptMeas} contains the proof of the optimality of the measurement scheme presented in Sec.~\ref{sec:Measurements}.

\section{The Model}
\label{sec:model}
Let us consider a set of $M$ bosonic modes described by the operators $\hat{a}_m$ and $\hat{a}_m^\dag$ satisfying the canonical commutation relations 
\begin{equation} [\hat{a}_m, \hat{a}_n ]=0, \qquad  
[\hat{a}_m, \hat{a}_n^\dag ]=\delta _{mn}  \qquad (m,n=1,\ldots,M).\label{CANO} 
\end{equation}  
The passive Gaussian unitary  $\hat{U}_\varphi$ of Fig.~\ref{fig:BlackBox}   is defined by
 the mapping~\cite{ref:ContVarQI,ref:ContVarQI-Adesso}
\begin{equation}
\hat{U}_\varphi^\dag\hat{a}_m\hat{U}_\varphi=\sum_{n=1}^M(U_\varphi)_{mn}\hat{a}_n
\qquad(m=1,\ldots,M),
\label{eqn:UnitaryRotation}
\end{equation}
or simply written as $\hat{U}_\varphi^\dag\hat{\bm{a}}\hat{U}_\varphi=U_\varphi\hat{\bm{a}}$ with $\hat{\bm{a}}
=
(\ %
\hat{a}_1
\ \cdots\ %
\hat{a}_M
\ )^T$, where $U_\varphi$ is an $M\times M$ unitary matrix, whose functional dependence upon $\varphi$ is assumed to be smooth. 
We remind that this kind of transformation
preserves the total number of photons of the system, i.e. 
\begin{equation}
\hat{U}_\varphi^\dag\hat{N}\hat{U}_\varphi=\hat{N}, \qquad 
\hat{N}=\sum_{m=1}^M\hat{a}_m^\dag\hat{a}_m,
\label{eqn:N}
\end{equation}
and can be constructed by using beam splitters and phase shifters.

Our problem is to estimate the actual value of the parameter $\varphi$  embedded in $\hat{U}_\varphi$ by probing the output state $\hat{\rho}_\varphi$ in~(\ref{OUTPUT}). Consider hence 
 a generic positive operator-valued measure (POVM) $\mathcal{P}=\{\hat{\Pi}_s\}_s$~\cite{ref:NielsenChuang,ref:HayashiIshizukaKawachiKimuraOgawa}  producing measurement outcomes $s$ with probabilities
\begin{equation} \label{defp} 
p(s|\varphi)=\Tr(\hat{\Pi}_s\hat{\rho}_\varphi).
\end{equation}
The  Cram\'er-Rao inequality~\cite{ref:Helstrom,ref:BraunsteinCave1994,ref:BraunsteinCave1996AnnPhys,ref:MetrologyScience,ref:HayashiAsymptoticTheory,ref:Paris-IJQI,ref:MetrologyNaturePhoto,ref:HolevoSNS} establishes that any attempt at estimating  $\varphi$  from the values of $s$ is characterized by an uncertainty 
\begin{equation} \label{CRCCC} 
\delta\varphi \ge\frac{1}{\sqrt{F(\varphi|\mathcal{P},\hat{\rho})}},
\end{equation}
with 
\begin{equation}
F(\varphi|\mathcal{P},\hat{\rho})=\sum_sp(s|\varphi)\left(
\frac{\partial}{\partial\varphi}\ln p(s|\varphi)
\right)^2
\label{eqn:FI}
\end{equation}
being the Fisher information (FI) of the process. 
A stronger, universal bound on the attainable estimation error can now be obtained by optimizing the right-hand side of~(\ref{CRCCC}) 
with respect to all possible POVMs $\mathcal{P}$.
This yields the 
quantum Cram\'er-Rao inequality~(\ref{CRQQQQ}), with  
\begin{equation}
F(\varphi|\hat{\rho}) = \max_\mathcal{P} F(\varphi|\mathcal{P}, \hat{\rho})
\label{boundonF}
\end{equation}
being the QFI of the problem, which
by construction depends only upon  the input state $\hat{\rho}$ and the circuit $\hat{U}_\varphi$~\cite{ref:Helstrom,ref:BraunsteinCave1994,ref:BraunsteinCave1996AnnPhys,ref:MetrologyScience,ref:HayashiAsymptoticTheory,ref:Paris-IJQI,ref:MetrologyNaturePhoto,ref:HolevoSNS}. 
The maximization in~(\ref{boundonF}) can be analytically solved, yielding the following compact expression 
\begin{equation} \label{QFI} 
 F(\varphi|\hat{\rho})
=\Tr(\hat{\rho}_\varphi\hat{L}_\varphi^2),
\end{equation}
with  $\hat{L}_\varphi$ being a Hermitian operator called symmetric logarithmic derivative (SLD), satisfying 
\begin{equation}
\frac{\rmd\hat{\rho}_\varphi}{\rmd\varphi}
=\frac{1}{2}(
\hat{L}_\varphi\hat{\rho}_\varphi
+\hat{\rho}_\varphi\hat{L}_\varphi
).
\label{eqn:LyapunovEq}
\end{equation}

The goal of the present work is to optimize the value of the QFI $F(\varphi|\hat{\rho})$ in~(\ref{QFI}) with respect to a special class of allowed input states $\hat{\rho}$. 
In particular, we shall restrict the analysis to the set $\mathcal{G}(M,\overline{N})$  of $M$-mode Gaussian states with a fixed average photon number $\overline{N}$, i.e. 
\begin{equation} \label{NUMB} 
\Tr(\hat{N} \hat{\rho})= \overline{N}.
\end{equation} 
This last condition is  motivated by the fact that it is not realistic to consider  probing signals with unbounded input energy. 
It turns out that for generic (non-Gaussian) input states the constraint~(\ref{NUMB}) is not
strong enough to keep the QFI $F(\varphi|\hat{\rho})$ finite
[see for instance Ref.~\cite{ref:EstId}, where, for the case with $M=2$ input modes, obtaining finite optimal values for $F(\varphi|\hat{\rho})$ requires to impose an extra condition on the variance of $\hat{N}$ on $\hat{\rho}$; see also Ref.~\cite{ref:RivasLuis-NJP2012}], yet for Gaussian inputs this suffices and the QFI $F(\varphi|\hat{\rho})$ is finite under the constraint~(\ref{NUMB}).

\section{Optimization of QFI}
\label{sec:Optimization}
As recapitulated in \ref{sec:Setup},  an input state $\hat{\rho}$ belonging to the  Gaussian set $\mathcal{G}(M,\overline{N})$  is fully characterized by a ($2M \times 2 M$ real, symmetric, and positive-definite) covariance matrix $\Gamma$ with matrix elements
\begin{equation}
\Gamma _{mn}=\frac{1}{2}\langle\{\hat{z}_m,\hat{z}_n\}\rangle
- \langle\hat{z}_m\rangle\langle\hat{z}_n\rangle
\qquad
(m,n=1,\ldots,2M)
\label{eqn:Gamma}
\end{equation}
and a ($2M$ real column) displacement vector 
\begin{equation}
\bm{d}=\langle\hat{\bm{z}}\rangle
\label{eqn:D}
\end{equation}
which satisfy the constraint~(\ref{NUMB}), i.e.
\begin{equation}
\frac{1}{2}\Tr\!\left(
\Gamma-\frac{1}{2}
\right)
+ \frac{1}{2} \bm{d}^2
= \overline{N},
\label{eqn:constraint}
\end{equation}
where $\hat{\bm{z}}=(\ \hat{\bm{x}}\ \ \hat{\bm{y}}\ )^T$ is the quadrature operator vector with $\hat{x}_m=(\hat{a}_m+\hat{a}_m^\dag)/\sqrt{2}$ and $\hat{y}_m=(\hat{a}_m-\hat{a}_m^\dag)/\sqrt{2}i$ ($m=1,\ldots ,M$), and $\langle{}\cdots{}\rangle$ denotes the expectation value on $\hat{\rho}$.
Furthermore, since $\hat{U}_\varphi$ is a passive Gaussian unitary, the associated output state $\hat{\rho}_\varphi$ obtained as~(\ref{OUTPUT}) also belongs to $\mathcal{G}(M,\overline{N})$, and its covariance matrix  $\Gamma_\varphi$ and displacement
vector  $\bm{d}_\varphi$ depend linearly on $\Gamma$ and $\bm{d}$, as
\begin{equation}
\Gamma_\varphi=R_\varphi\Gamma R_\varphi^T
,\qquad
\bm{d}_\varphi=R_\varphi\bm{d},
\label{eqn:RotGauss}
\end{equation}
where $R_\varphi$ is the orthogonal matrix rotating the quadrature operators according to $\hat{U}_\varphi$ (see \ref{sec:Setup}).
Under this condition, the SLD fulfilling~(\ref{eqn:LyapunovEq})  can  be expressed as~\cite{arXiv:1303.3682} 
\begin{equation}
\hat{L}_\varphi
=
(
\hat{\bm{z}}
-\bm{d}_\varphi
)^T
\Lambda_\varphi
(
\hat{\bm{z}}
-\bm{d}_\varphi
)
+\frac{\partial\bm{d}_\varphi^T}{\partial\varphi}
\Gamma_\varphi^{-1}
(
\hat{\bm{z}}
-\bm{d}_\varphi
)
-\Tr(\Lambda_\varphi\Gamma_\varphi),
\label{eq:SLDGauss}
\end{equation}
and, accordingly, the 
 QFI reads~\cite{arXiv:1303.3682,ref:GaussFidelity}
\begin{equation}
F(\varphi|\hat{\rho})
=\Tr\!\left(
\Lambda_\varphi
\frac{\partial\Gamma_\varphi}{\partial\varphi}
\right)
+\frac{\partial\bm{d}_\varphi^T}{\partial\varphi}
\Gamma_\varphi^{-1}
\frac{\partial\bm{d}_\varphi}{\partial\varphi}.
\label{eqn:QFIMixGauss}
\end{equation}
Here, $\Lambda_\varphi$ is the solution to
\begin{equation}
iJ\Lambda_\varphi
-(2\Gamma_\varphi iJ)^{-1}
iJ\Lambda_\varphi
(2\Gamma_\varphi iJ)^{-1}
=-\frac{\partial(2\Gamma_\varphi iJ)^{-1}}{\partial\varphi},
\label{eqn:LambdaEq}
\end{equation}
with $J$ being the $2M \times 2M$ matrix
\begin{equation}
J=
\left(\begin{array}{@{}c|c@{}}
\,\,0&\openone\,\,\\
\hline
\,\,-\openone & 0\,\,
\end{array}\right),
\label{DEFJJJqui} 
\end{equation}
known as the symplectic form.

In the remainder of this section, we shall employ these expressions to derive the inequality~(\ref{eqn:OptQFISeq}).
The analysis will be split into two parts, addressing first the case of the pure elements of $\mathcal{G}(M,\overline{N})$ and then the case of the mixed ones. 
  For those who are familiar with QFI optimization problems, this procedure might sound unecessary. 
  Indeed, due to the convexity of QFI~\cite{ref:QFI-convexity-Fujiwara2001,ref:MetrologyNaturePhoto}, it is well known that 
pure input states perform better than mixed input states for metrological purposes.
We cannot, however, apply the same argument in the present case, and it is not obvious at first glance whether the best state is a pure state.
Indeed, even though it is true that any mixed Gaussian state can be decomposed as a convex sum of pure Gaussian states, each of the constituent of such decomposition does not necessarily satisfy the constraint~(\ref{NUMB}) on the photon number in general.
In short, the Gaussian set $\mathcal{G}(M,\overline{N})$ is not a convex set, and therefore we cannot use the convexity argument to optimize the QFI\@.
As a consequence, for the problem we are considering here, we have to address explicitly the case of mixed input states.

\subsection{Optimization among Pure Gaussian Inputs}
For a pure Gaussian state $|\psi\rangle\in\mathcal{G}(M,\overline{N})$, the symplectic eigenvalues of its covariance matrix $\Gamma$ [i.e.~the parameters $\{\sigma_1,\ldots,\sigma_M\}$ in the canonical decomposition~(\ref{eqn:SymplecticDecomp}) of $\Gamma$] are all equal to $\sigma_m=1/2$ ($m=1,\ldots,M$). 
Accordingly,
introducing a $2M\times 2M$ symplectic orthogonal matrix $R$ (i.e.~an orthogonal matrix $R$ satisfying $R^TJR=J$) and  an $M\times M$ diagonal positive matrix $r$, the covariance matrix $\Gamma$ can be decomposed as [see the canonical decomposition~(\ref{eqn:SymplecticDecomp}) of $\Gamma$ of a generic (mixed) Gaussian state]
\begin{equation}
\Gamma
=\frac{1}{2}RQ^2R^T
=\frac{1}{2}R\left(\begin{array}{@{}c|c@{}}
\,\,\,e^{r} & 0\,\, \\
\hline
\,\,\,0 & e^{-r}\,\,
\end{array}\right)^2R^T,
\label{eqn:GammaPure}
\end{equation}
while the constraint~(\ref{eqn:constraint}) on the average number becomes
\begin{equation}
\Tr\sinh^2r+\frac{1}{2}\bm{d}^2
=\overline{N},
\label{eqn:constraintpure}
\end{equation}
with $\bm{d}$ being the displacement vector of $|\psi\rangle$.

Exactly the  same properties 
hold for the covariance matrix $\Gamma_\varphi$ and the displacement 
$\bm{d}_\varphi$ of the associated  output counterpart~(\ref{OUTPUT}) of $|\psi\rangle$, which of course is also a pure element of the set $\mathcal{G}(M,\overline{N})$. 
Under this premise, the equation~(\ref{eqn:LambdaEq}) for $\Lambda_\varphi$ can be solved explicitly, yielding
\begin{equation} 
\Lambda_\varphi
=
-\frac{1}{4}\frac{\partial\Gamma_\varphi^{-1}}{\partial\varphi} \label{EXACT}.
\end{equation} 
The QFI~(\ref{eqn:QFIMixGauss}) is then reduced to~\cite{arXiv:1303.3682}
\begin{equation}
F(\varphi|\hat{\rho})
=\frac{1}{4}\Tr\!\left[
\left(
\Gamma_\varphi^{-1}
\frac{\partial\Gamma_\varphi}{\partial\varphi}
\right)^2
\right]
+\frac{\partial\bm{d}_\varphi^T}{\partial\varphi}
\Gamma_\varphi^{-1}
\frac{\partial\bm{d}_\varphi}{\partial\varphi},
\label{eqn:QFIpure}
\end{equation}
with the SLD~(\ref{eq:SLDGauss}) given by
\begin{eqnarray}
\hat{L}_\varphi
={}&
{-\frac{1}{4}}
(
\hat{\bm{z}}
-\bm{d}_\varphi
)^T
\frac{\partial\Gamma_\varphi^{-1}}{\partial\varphi}
(
\hat{\bm{z}}
-\bm{d}_\varphi
)
\nonumber\\
&{}
+\frac{\partial\bm{d}_\varphi^T}{\partial\varphi}
\Gamma_\varphi^{-1}
(
\hat{\bm{z}}
-\bm{d}_\varphi
)
-\frac{1}{4}
\Tr\!\left(
\Gamma_\varphi^{-1}
\frac{\partial\Gamma_\varphi}{\partial\varphi}
\right).
\label{eq:SLDpureGauss}
\end{eqnarray}
A further simplification can then be obtained by invoking~(\ref{eqn:RotGauss}), which expresses the functional dependence of $\Gamma_\varphi$ and $\bm{d}_\varphi$ in terms of the symplectic orthogonal matrix $R_\varphi$ representing the passive Gaussian unitary transformation $\hat{U}_\varphi$. Specifically, we get 
\begin{equation}
F(\varphi|\hat{\rho})=\frac{1}{2}\Tr(G_\varphi  \Gamma^{-1}G_\varphi \Gamma-G_\varphi^2 )+\bm{d}^T G_\varphi \Gamma^{-1} G_\varphi \bm{d},
\label{eqn:QFIlinear}
\end{equation}
and 
\begin{equation}
\hat{L}_\varphi
=
\frac{i}{4}
(
R_\varphi^T\hat{\bm{z}}
-\bm{d}
)^T
[G_\varphi,\Gamma^{-1}]
(
R_\varphi^T\hat{\bm{z}}
-\bm{d}
)
+i
\bm{d}^TG_\varphi
\Gamma^{-1}
(
R_\varphi^T\hat{\bm{z}}
-\bm{d}
),
\label{eqn:SLDlinear}
\end{equation}
where
\begin{equation}
G_\varphi=iR_\varphi^{T}\frac{\rmd R_\varphi}{\rmd \varphi}
\end{equation}
is the generator of $R_\varphi$.

Our problem is, therefore, to  maximize the QFI $F(\varphi|\hat{\rho})$ in~(\ref{eqn:QFIlinear}) with respect to  $\Gamma$  and $\bm{d}$, keeping in mind the parametrization~(\ref{eqn:GammaPure})
and the constraint~(\ref{eqn:constraintpure}). 
For this purpose, 
we start bounding the first term $F^{(1)}(\varphi|\hat{\rho})$ in the sum~(\ref{eqn:QFIlinear}).
By plugging the symplectic decomposition~(\ref{eqn:GammaPure}) of $\Gamma$, and using the parameterization~(\ref{eqn:R}) for $R$ as well as the structure~(\ref{eqn:G}) of the generator $G_\varphi$, we get (see \ref{app:3.17} for the derivation)
\begin{eqnarray}
F^{(1)}(\varphi|\hat{\rho})
&=&\frac{1}{2}\Tr(
G_\varphi\Gamma^{-1}G_\varphi\Gamma-G_\varphi^2
)\nonumber\\
&=&
\Tr[
(U^\dag g_\varphi U\cosh 2r) ^2
]
-\Tr(g_\varphi^2)
\nonumber\\
&&{}
+\Tr(
U^\dag g_\varphi U\sinh 2r\,U^T g_\varphi^*U^*\sinh 2r 
),
\label{eqn:F1}
\end{eqnarray}
where $g_\varphi$ is the generator of the unitary matrix $U_\varphi$ as introduced in (\ref{eqn:Hamiltonian0}) and involved in the structure of $G_\varphi$ in (\ref{eqn:G}), while $U$ is the unitary matrix appearing in the parametrization of $R$ in (\ref{eqn:R}).
This quantity can be bounded from above as
\begin{eqnarray}
F^{(1)}(\varphi|\hat{\rho})
&\le&
\Tr[
(U^\dag g_\varphi U)^2
\cosh^22r
]
+
\Tr[
(U^\dag g_\varphi U)^2
\sinh^22r
]
-\Tr(g_\varphi^2)
\nonumber\\
&=&2
\Tr[
(U^\dag g_\varphi U)^2
\sinh^22r
]
\nonumber\\
&\le&2
\|g_\varphi\|^2\Tr\sinh^22r,
\label{eqn:IneqF1}
\end{eqnarray}
where we have used the inequalities
\begin{eqnarray}
\Tr[(AB)^2]
\le\Tr(A^2B^2),
\label{eqn:Ineq11}\\
\Tr(A^TB^TAB)
\le\Tr(A^2B^2),
\label{eqn:Ineq22}
\end{eqnarray}
valid for Hermitian matrices $A$ and $B$, and
\begin{equation}
\Tr(AB)\le\|A\|\Tr B,
\end{equation}
valid for Hermitian and positive semi-definite matrices $A$ and $B$ (see \ref{app:Ineq} for their proofs).
Note that $g_\varphi$ is Hermitian and hence $(U^\dag g_\varphi U)^2$ is positive semi-definite, and its norm is given by $\|(U^\dag g_\varphi U)^2\|=\|g_\varphi\|^2$.
The equality in~(\ref{eqn:Ineq11}) holds if and only if $[A,B]=0$, while the equality in~(\ref{eqn:Ineq22}) is obtained if and only if $AB=(AB)^T$.

The second term $F^{(2)}(\varphi|\hat{\rho})$ in~(\ref{eqn:QFIlinear}), on the other hand, can be  bounded from above as
\begin{eqnarray}
F^{(2)}(\varphi|\hat{\rho})
&=&\bm{d}^T G_\varphi\Gamma^{-1}G_\varphi\bm{d} 
\nonumber\\
&=&
2\bm{d}^T G_\varphi RQ^{-2}R^TG_\varphi\bm{d} 
\nonumber\\
&\le&
2\|G_\varphi RQ^{-2}R^TG_\varphi\|\bm{d}^2
\nonumber\\
&\le&
2\|G_\varphi\|^2\|Q^{-2}\|\bm{d}^2
\nonumber\\
&=&2\|g_\varphi\|^2\|e^{2r}\|\bm{d}^2,
\label{eqn:F2opt}
\end{eqnarray}
where we have assumed, without loss of generality, that $r_m\ge0$ ($m=1,\ldots,M$).

Exploiting these results, we can then bound the QFI~(\ref{eqn:QFIlinear}) as 
\begin{eqnarray}
F(\varphi|\hat{\rho})
&\le&2\|g_\varphi\|^2
(
\Tr
\sinh^22r
+
\|e^{2r}\|
\bm{d}^2
)
\nonumber\\
&=&2\|g_\varphi\|^2
(
4\Tr
\sinh^2r
+4
\Tr
\sinh^4r
+
\|e^{2r}\|\bm{d}^2
)
\nonumber\\
&\le&2\|g_\varphi\|^2
[
4\Tr
\sinh^2r
+4(\Tr\sinh^2r)^2
+2\|{\cosh2r}\|
\bm{d}^2
]\nonumber\\
&=&2\|g_\varphi\|^2
[
4\Tr
\sinh^2r
+4(\Tr\sinh^2r)^2
+
(4\|{\sinh^2r}\|+2)
\bm{d}^2
]\nonumber\\
&\le&2\|g_\varphi\|^2
[
4\Tr
\sinh^2r
+4(\Tr\sinh^2r)^2
+
(4\Tr\sinh^2r+2)
\bm{d}^2
],
\label{eqn:Fopt}
\end{eqnarray}
where  we have used the inequality 
\begin{equation}
\Tr(A^2)\le(\Tr A)^2,
\end{equation}
valid for a positive semi-definite matrix $A$, which is saturated if and only if only one of the eigenvalues of $A$ is nonvanishing and it is not degenerate (see \ref{app:Ineq} for its proof).
Imposing hence the  constraint~(\ref{eqn:constraintpure}), this finally gives us 
\begin{eqnarray} 
F(\varphi|\hat{\rho}) 
&\leq&8\|g_\varphi\|^2
\left(
\overline{N}(\overline{N}+1)
-\frac{1}{4}\bm{d}^4
\right)
\le8\|g_\varphi\|^2
\overline{N}(\overline{N}+1),
\label{eqn:Fopt1}
\end{eqnarray}
which proves the inequality~(\ref{eqn:OptQFISeq}) for the case of  pure input Gaussian states.  
This result reproduces the bounds previously known for $M=1$ (single-mode phase shift)~\cite{PhysRevA.73.033821,PhysRevA.79.033834,PhysRevA.88.040102,PhysRevA.91.013808,PhysRevA.94.062313} and for $M=2$ (general two-mode passive linear circuits)~\cite{ref:EstId}, and generalizes them to $M\ge3$.

\subsubsection{Optimal States} 
The above derivation of the bound not only  proves  that the inequality~(\ref{eqn:OptQFISeq}) holds at least for the pure input states of the set $\mathcal{G}(M,\overline{N})$, but also
that the bound is saturated by a proper choice of the inputs, i.e.~by properly tuning the parameters in $\Gamma$ and $\bm{d}$. 
Let us identify such input states.
\begin{enumerate}
\item
In order to saturate the last inequality in~(\ref{eqn:Fopt1}), the necessary and sufficient condition is
\begin{equation}
\bm{d}=0.
\label{eqn:d0}
\end{equation}
\item 
Then, the last inequality in~(\ref{eqn:Fopt}) is automatically saturated, and the second inequality in~(\ref{eqn:Fopt}) is saturated if and only if only one (e.g.~the first) of the squeezing parameters $\{r_1,\ldots,r_M\}$ of the matrix $r$ is nonvanishing.
Let us put the nonvanishing squeezing parameter $r_0\,(>0)$ in the first mode, 
\begin{equation}
r=
\left(\begin{array}{cccc}
r_0&&&\\
&0&&\\[-1truemm]
&&\ddots&\\
&&&0
\end{array}\right).
\label{eqn:ropt}
\end{equation}
\item
The equality in~(\ref{eqn:F2opt}) is trivially satisfied, since $\bm{d}$ is required to be vanishing in~(\ref{eqn:d0}).
\item The last inequality in~(\ref{eqn:IneqF1}) is saturated if and only if the vector $(1\ \ 0\ \ \cdots\ \ 0)^T$, corresponding to the first mode, belongs to the eigenspace of $(U^\dag g_\varphi U)^2$ associated with its largest eigenvalue.
The choice
\begin{equation}
U=V_\varphi,
\label{eqn:UV}
\end{equation}
with $V_\varphi$ introduced in~(\ref{eqn:Hamiltonian11}) to diagonalize $g_\varphi$, suffices to fulfill this condition. 
Note that the eigenvalues $\{\varepsilon_1,\ldots,\varepsilon_M\}$ of $g_\varphi$ in~(\ref{eqn:Hamiltonian11}) are ordered in descending order in their magnitudes.
\item The first inequality in~(\ref{eqn:IneqF1}) is saturated if and only if both conditions
\begin{equation}
\left\{
\begin{array}{l}
\medskip
[U^\dag g_\varphi U,\cosh2r]=0,\\
U^\dag g_\varphi U
\sinh 2r
=(U^\dag g_\varphi U
\sinh 2r
)^T
\end{array}\right.
\label{eqn:CondEqsF1}
\end{equation}
are satisfied: recall the conditions for the equalities in~(\ref{eqn:Ineq11}) and~(\ref{eqn:Ineq22}).
These conditions are already satisfied with the above tunings of $r$ and $U$ in~(\ref{eqn:ropt}) and~(\ref{eqn:UV}).
\item
Finally, since $\bm{d}=0$, all the photons are spent for the squeezing $r_0$ in the first mode.
The constraint on the mean photon number $\overline{N}$ in~(\ref{eqn:constraintpure}) yields
\begin{equation}
r_0=\ln\!\left(\sqrt{\overline{N}}+\sqrt{\overline{N}+1}\right).
\label{eqn:r0}
\end{equation}
\end{enumerate}

Putting all these conditions together, it follows  that the state achieving the upper bound in~(\ref{eqn:Fopt1}) among the pure Gaussian input states of $\mathcal{G}(M,\overline{N})$ is a \textit{single-mode squeezed vacuum} with zero displacement~(\ref{eqn:d0}) and a squeezing $r$ given by~(\ref{eqn:ropt}) and (\ref{eqn:r0}), and rotated by the unitary~(\ref{eqn:UV}), i.e.~the vector 
\begin{equation}
\ket{\psi_\mathrm{opt}}
=\hat{V}_\varphi\hat{S}_1(r_0)\ket{0},
\label{eqn:OptPure}
\end{equation}
with  $\ket{0}$ the vacuum state and
$
\hat{S}_1(\xi)=e^{\frac{1}{2}(\xi\hat{a}_1^{\dag2}-\xi^*\hat{a}_1^2)}
$
the squeezing operator on the first mode.

A couple of comments are in order.
First, recall that $\hat{V}_\varphi$ is the passive linear transformation characterized by $\hat{V}_\varphi^\dag\hat{\bm{a}}\hat{V}_\varphi=V_\varphi\hat{\bm{a}}$ with the $M\times M$ unitary matrix $V_\varphi$  diagonalizing the generator $g_\varphi$ of the circuit as in~(\ref{eqn:Hamiltonian11}).
It redefines the modes of the system in a way that allows us to describe the optimal state as a configuration with all the photons injected into the first mode only [i.e.~the one with the largest (in magnitude) eigenvalue of $g_\varphi$]. 
We stress, however, that even after this ``reorganization'' the modes other than the first one are not free from the target parameter $\varphi$ in general, due to the subsequent propagation induced by $\hat{U}_\varphi$, and the problem is not reduced to a single-mode problem. It remains intrinsically a multimode problem, and we cannot simply apply the results known for single-mode estimation problems.
Second, as indicated by the notation, the transformation $\hat{V}_\varphi$ may depend upon the target parameter $\varphi$ for a generic choice of $\hat{U}_\varphi$, and so may do the optimal state $|\psi_\mathrm{opt}\rangle$. 
Therefore, if that is the case, it would not be easy to prepare this optimal state $|\psi_\mathrm{opt}\rangle$ without knowing the value of the parameter $\varphi$, which we intend to estimate, and an adaptive strategy updating the estimate of $\varphi$ would be required in practice.

\subsection{Optimization among Mixed Gaussian Inputs}
We have just shown that the inequality~(\ref{eqn:OptQFISeq}) holds at least for the pure elements of the set $\mathcal{G}(M,\overline{N})$.
Here, we are going to generalize this by showing that the same result holds for the mixed elements of the set $\mathcal{G}(M,\overline{N})$.

We first point out that any mixed Gaussian state $\hat{\rho}_{\Gamma,\bm{d}}$, characterized by a covariance matrix $\Gamma$ and a displacement $\bm{d}$, can be expressed as a mixture of pure Gaussian states $\hat{\rho}_{\Gamma_0,\bm{d}-\bm{\xi}}$ as 
\begin{equation}
\hat{\rho}_{\Gamma,\bm{d}}
=\int \rmd^{2M}\bm{\xi}\,P_\Gamma(\bm{\xi})\hat{\rho}_{\Gamma_0,\bm{d}-\bm{\xi}}
\label{eqn:ConvexDecomp}
\end{equation}
with a Gaussian probability distribution
\begin{equation}
P_\Gamma(\bm{\xi})=\frac{e^{-\frac{1}{2}\bm{\xi}^T(\Gamma-\Gamma_0)^{-1}\bm{\xi}}}{\sqrt{(2\pi)^{2M}\det(\Gamma-\Gamma _0)}}.
\label{eqn:GaussianWeight}
\end{equation}
In these expressions,  $\Gamma_0$ is the pure-state covariance matrix obtained by taking the symplectic decomposition~(\ref{eqn:SymplecticDecomp}) of the original
covariance matrix $\Gamma$ and 
replacing all the symplectic eigenvalues $\{\sigma_1,\ldots,\sigma_M\}$ of the latter with $1/2$, i.e. 
\begin{equation}
\Gamma_0=\frac{1}{2}RQ^2R^T,
\end{equation}
keeping the squeezing matrix $Q$ and the symplectic orthogonal matrix $R$ of $\Gamma$ unchanged.
By construction, it follows that 
\begin{equation} 
\Gamma-\Gamma_0\ge0, \label{IMPO} \end{equation} 
since all the symplectic eigenvalues $\{\sigma_1,\ldots,\sigma_M\}$ of any $\Gamma$ are greater than or equal to $1/2$.
The convex decomposition~(\ref{eqn:ConvexDecomp}) can be verified by looking at the characteristic function $\chi_{\Gamma,\bm{d}}(\bm{\eta})$ for the Gaussian state $\hat{\rho}_{\Gamma,\bm{d}}$ in~(\ref{eqn:CharFunc}): by direct computation, we can check that
\begin{equation}
\int \rmd^{2M}\bm{\xi}\,P_\Gamma(\bm{\xi})\chi_{\Gamma_0,\bm{d}-\bm{\xi}}(\bm{\eta})
=\chi_{\Gamma,\bm{d}}(\bm{\eta}),
\end{equation}
which is equivalent to (\ref{eqn:ConvexDecomp}).
Note that the pure Gaussian states $\hat{\rho}_{\Gamma_0,\bm{d}-\bm{\xi}}$ in the convex sum~(\ref{eqn:ConvexDecomp}) do \textit{not} satisfy the constraint~(\ref{eqn:constraint})
 on the mean photon number in general, while the original mixed state $\hat{\rho}_{\Gamma,\bm{d}}$ should do.
Yet, by using the convexity of the QFI and the last inequality appearing in~(\ref{eqn:Fopt}), which holds for pure Gaussian states, and by recalling the expressions for the mean photon number in~(\ref{eqn:constraint}) and~(\ref {eqn:constraintpure}), we can write 
\begin{eqnarray}
\fl
F(\varphi|\hat{\rho}_{\Gamma,\bm{d}})
\leq
\int \rmd^{2M}\bm{\xi}\,P_\Gamma(\bm{\xi})
F(\varphi|\hat{\rho}_{\Gamma_0,\bm{d}-\bm{\xi}})
\nonumber\\
\fl\hphantom{F(\varphi|\hat{\rho}_{\Gamma,\bm{d}})}
\le8\|g_\varphi\|^2
\int \rmd^{2M}\bm{\xi}\,P_\Gamma(\bm{\xi})\,\Biggl\{
\frac{1}{2}
\left[
\Tr\!\left(
\Gamma_0-\frac{1}{2}
\right)
+(\bm{d}-\bm{\xi})^2
\right]
\nonumber\\
\fl\hphantom{F(\varphi|\hat{\rho}_{\Gamma,\bm{d}})\le{}}
\qquad\qquad\qquad\qquad\quad\ \ %
{}+
\frac{1}{4}
\left[
\Tr\!\left(
\Gamma_0-\frac{1}{2}
\right)
+(\bm{d}-\bm{\xi})^2
\right]^2
-\frac{1}{4}(\bm{d}-\bm{\xi})^4
\Biggr\}
\nonumber\\
\fl\hphantom{F(\varphi|\hat{\rho}_{\Gamma,\bm{d}})}
=8\|g_\varphi\|^2
\,\Biggl\{
\frac{1}{2}
\left[
\Tr\!\left(
\Gamma-\frac{1}{2}
\right)
+\bm{d}^2
\right]
\nonumber\\
\fl\hphantom{F(\varphi|\hat{\rho}_{\Gamma,\bm{d}})\le{}}
\qquad\quad\ \ %
{}
+
\frac{1}{4}
\left[
\Tr\!\left(
\Gamma-\frac{1}{2}
\right)
+\bm{d}^2
\right]^2
-
\frac{1}{4}
[\Tr(\Gamma-\Gamma_0)+\bm{d}^2]^2
\Biggr\}
\nonumber\\
\fl\hphantom{F(\varphi|\hat{\rho}_{\Gamma,\bm{d}})}
=8\|g_\varphi\|^2\left(
\overline{N}(\overline{N}+1)
-\frac{1}{4}[
\Tr(\Gamma-\Gamma_0) +\bm{d}^2]^2\right)
\nonumber\\
\fl\hphantom{F(\varphi|\hat{\rho}_{\Gamma,\bm{d}})}
\le8\|g_\varphi\|^2\overline{N}(\overline{N}+1).
\vphantom{\frac{1}{2}}
\label{eqn:BoundQFIMixed}
\end{eqnarray}
For the first equality, we have used the moments of the Gaussian distribution $P_\Gamma(\bm{\xi})$ in (\ref{eqn:GaussianWeight}), i.e., $\int\rmd^{2M}\bm{\xi}\,P_\Gamma(\bm{\xi})\bm{\xi}=0$ and $\int\rmd^{2M}\bm{\xi}\,P_\Gamma(\bm{\xi})\bm{\xi}^2=\Tr(\Gamma-\Gamma_0)$.
The inequality~(\ref{eqn:BoundQFIMixed}) proves that~(\ref{eqn:OptQFISeq}) holds irrespective of the purity of the input states.
Furthermore, we notice that 
 the last inequality is saturated if and only if 
\begin{equation}
\Tr(\Gamma-\Gamma_0)=0\quad\mathrm{and}\quad
\bm{d}=0.
\end{equation}
Due to~(\ref{IMPO}), the first condition requires 
\begin{equation}
\Gamma=\Gamma_0,
\end{equation}
implying that the only elements of $\mathcal{G}(M,\overline{N})$ which saturate the bound~(\ref{eqn:OptQFISeq}) are the pure ones, given in~(\ref{eqn:OptPure}).

\section{Measurements}
\label{sec:Measurements}
In this section, we focus on the measurement $\mathcal{P}$ that attains the maximum on the right-hand side of~(\ref{boundonF}) yielding the 
QFI\@.  
As it is the case for the optimal input state $|\psi_\mathrm{opt}\rangle$ 
analyzed in the previous section, we shall see that the optimal POVM also exhibits in general a nontrivial  dependence on the 
target parameter $\varphi$, making it problematic to use it in realistic situations. 
Still, determining the optimal POVM explicitly is a well-defined  problem which deserves to be addressed.

As a starting point of our study, we use the well-known fact that a POVM $\mathcal{P}$ that maximizes the FI of the problem can always be constructed by looking at the set of the eigenprojections of the SLD $L_\varphi$ of the model~\cite{ref:Paris-IJQI}.
We have given an SLD $L_\varphi$ for a generic Gaussian state $\hat{\rho}_\varphi$ in~(\ref{eq:SLDGauss}), which for a pure Gaussian state reduces to~(\ref{eq:SLDpureGauss}).
For our problem, in which the parameter $\varphi$ is embedded in the probe state via a passive linear circuit, it reduces further to~(\ref{eqn:SLDlinear}), which depends on the input state, i.e.~its covariance matrix $\Gamma$ and displacement $\bm{d}$, and the generator $G_\varphi$ of the circuit.
Specifying this expression in the case of the optimal input $\ket{\psi_\mathrm{opt}}$ in~(\ref{eqn:OptPure}), we get 
\begin{equation}
\hat{L}_\varphi
=
i\varepsilon_1
\sinh2r_0
\,
\hat{U}_\varphi
\hat{V}_\varphi
(
\hat{a}_1^2
-\hat{a}_1^{\dag2}
)
\hat{V}_\varphi^\dag
\hat{U}_\varphi^\dag,
\label{eqn:SLDoptGauss}
\end{equation}
with $\hat{a}_1$ being the annihilation operator of the first probing mode, and $\varepsilon_1$ being the largest (in magnitude) eigenvalue of $G_\varphi$, which is put in the first mode after the diagonalization of $G_\varphi$ by $\hat{V}_\varphi$ [see~(\ref{eqn:G})--(\ref{eqn:DescendEpsilon}); recall also the discussion around~(\ref{eqn:UV})].

Notice, however, that SLD is not unique when the density operator $\hat{\rho}_\varphi$ is not of full rank: see~(\ref{eqn:LyapunovEq}).
Indeed, there is a different and simple construction of SLD for a pure state.
Since a pure state $\hat{\rho}_\varphi$ satisfies $\hat{\rho}_\varphi=\hat{\rho}_\varphi^2$, its derivative yields an SLD
$
\hat{L}_\varphi'
=2\,\rmd\hat{\rho}_\varphi/\rmd\varphi
$,
which for our problem with the optimal Gaussian input state $\ket{\psi_\mathrm{opt}}$ reads
\begin{equation}
\hat{L}_\varphi'
=-2i\hat{U}_\varphi
\left[ \hat{G}_\varphi, \ket{\psi_\mathrm{opt}} \bra{\psi_\mathrm{opt}} \right]
\hat{U}_\varphi^\dag,
\end{equation}
where
\begin{equation}
\hat{G}_\varphi=i\hat{U}_\varphi^\dag\frac{\rmd \hat{U}_\varphi}{\rmd \varphi}
\label{eqn:Ghat}
\end{equation}
is the generator of the target circuit $\hat{U}_\varphi$, which is quadratic in the canonical operators $\hat{\bm{a}}$ and $\hat{\bm{a}}^\dag$.
This SLD $\hat{L}_\varphi'$ is of rank 2, and its eigenbasis includes the two orthogonal eigenvectors
\begin{equation}
\ket{\phi_\pm}=\frac{1}{\sqrt{2}}\hat{U}_\varphi\,\Bigl(\ket{\psi_\mathrm{opt}}\mp i\ket{\psi_\mathrm{opt}^\perp}\Bigr)
\label{eqn:OptMeas1}
\end{equation}
belonging to the two nonvanishing eigenvalues $\pm2(\Delta{G}_\varphi)_\mathrm{opt}$, where
\begin{equation}
\ket{\psi_\mathrm{opt}^\perp}
=\frac{1}{(\Delta{G}_\varphi)_\mathrm{opt}}
\left(
\hat{G}_\varphi-\langle\hat{G}_\varphi\rangle_\mathrm{opt}
\right)\ket{\psi_\mathrm{opt}},
\end{equation}
with $\langle\hat{G}_\varphi\rangle_\mathrm{opt}=\bra{\psi_\mathrm{opt}}\hat{G}_\varphi\ket{\psi_\mathrm{opt}}$ and $(\Delta{G}_\varphi)^2_\mathrm{opt}=\langle\hat{G}_\varphi^2\rangle_\mathrm{opt}-\langle\hat{G}_\varphi\rangle_\mathrm{opt}^2$, is a state orthogonal to $\ket{\psi_\mathrm{opt}}$, i.e.~$\bracket{\psi_\mathrm{opt}}{\psi_\mathrm{opt}^\perp}=0$. 
Therefore, the measurement $\mathcal{P}$  with the POVM
\begin{equation}
\Bigl\{
\ket{\phi_+}\bra{\phi_+},\ket{\phi_-}\bra{\phi_-},\openone-\ket{\phi_+}\bra{\phi_+}-\ket{\phi_-}\bra{\phi_-}
\Bigr\}
\label{eqn:POVM1}
\end{equation}
will achieve the upper bound of the QFI in~(\ref{eqn:OptQFISeq}).
This is a generalization of the result given in Ref.~\cite{PhysRevA.73.033821}, from a single-mode phase shift to a  generic multimode passive linear circuit.

\begin{figure}[t]
\centering
\includegraphics[width=0.6\textwidth]{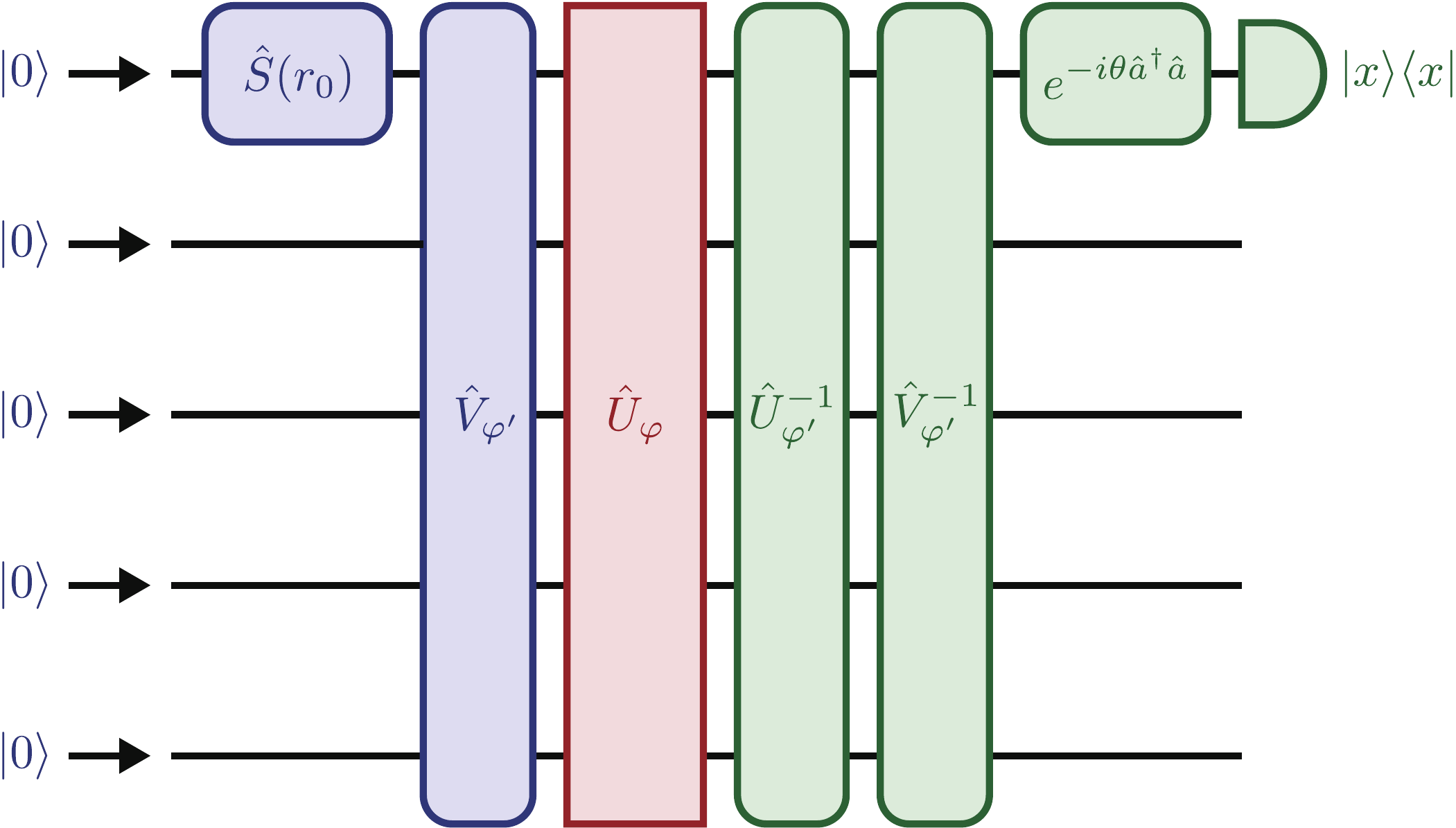}
\caption{An overall circuit to achieve the ultimate precision bound in~(\ref{eqn:OptQFISeq}), including the preparation stage for the optimal input state $\ket{\psi_\mathrm{opt}}$ in~(\ref{eqn:OptPure}) and the measurement stage, where $\ket{x}\bra{x}$ represents the homodyne measurement on the first mode along the quadrature $\hat{x}_1=(\hat{a}_1+\hat{a}_1^\dag)/\sqrt{2}$ and the phase shift $\theta$ is tuned to $\theta=\pm\tan^{-1}e^{2r_0}$.
Note that $\varphi$ of $\hat{U}_\varphi$ is the target parameter to be estimated, which is not under our control, while $\varphi'$ of $\hat{V}_{\varphi'}$, $\hat{U}_{\varphi'}^{-1}$, and $\hat{V}_{\varphi'}^{-1}$ is decided by ourselves. Tuning $\varphi'$ to the true value $\varphi$ provides us with the optimal strategy. The perfect cancellation of $\hat{U}_\varphi$ by $\hat{U}_{\varphi'}^{-1}$ tells us that our guessed value $\varphi'$ perfectly matches the true value $\varphi$, and a small deviation can be sensitively detected by the strategy shown here with $\varphi'=\varphi$.}
\label{fig:MeasCircuit}
\end{figure}
Another example of an optimal POVM can be obtained by considering the scheme depicted in Fig.~\ref{fig:MeasCircuit} [the circuit in Fig.~\ref{fig:MeasCircuit} includes both the preparation stage for the optimal input state $\ket{\psi_\mathrm{opt}}$ in~(\ref{eqn:OptPure}) and the probing stage together with the circuit  $\hat{U}_\varphi$].
The measurement is to first undo the circuit $\hat{U}_\varphi$ as well as the transformation $\hat{V}_\varphi$ applied to prepare the optimal input state $\ket{\psi_\mathrm{opt}}$ in~(\ref{eqn:OptPure}), and then to perform the homodyne measurement on the first mode along the quadrature $\hat{x}_1^{(\theta)}=e^{i\theta\hat{a}_1^\dag\hat{a}_1}\hat{x}_1e^{-i\theta\hat{a}_1^\dag\hat{a}_1}=\hat{x}_1\cos\theta+\hat{y}_1\sin\theta$ with $\theta=\pm\tan^{-1}e^{2r_0}$.
Accordingly, the elements $\{ \hat{\Pi}_x\}$ of the POVM for this measurement can be expressed as 
\begin{equation}
\hat{\Pi}_{x}
=\hat{U}_\varphi\hat{V}_\varphi e^{i\theta\hat{a}_1^\dag\hat{a}_1}\,\Bigl(
\ket{x}\bra{x}
\otimes\openone\otimes\cdots\otimes\openone\Bigr)\,e^{-i\theta\hat{a}_1^\dag\hat{a}_1}\hat{V}_\varphi^\dag\hat{U}_\varphi^\dag,
\label{eqn:HomodynePOVM}
\end{equation}
where $\ket{x}$ is the eigenvector of the quadrature operator $\hat{x}_1$ such that $\hat{x}_1\ket{x}=x\ket{x}$, 
normalized as $\bracket{x}{x'}=\delta(x-x')$.
Indeed, the FI by this POVM $\{\hat{\Pi}_x\}$ for the optimal input $\ket{\psi_\mathrm{opt}}$ in~(\ref{eqn:OptPure}) coincides with the upper bound of the QFI in~(\ref{eqn:OptQFISeq}).
See \ref{app:OptMeas} for the proof.
This is a generalization of the result given in Ref.~\cite{PhysRevA.79.033834}, from a single-mode phase shift to a generic multimode passive linear circuit.

\begin{figure}[b]
\begin{indented}
\item[]
\begin{tabular}{@{}l@{\qquad}l}
(a) MZ interferometer I
&
(b) MZ interferometer II
\\[2truemm]
\makebox(155,72){\includegraphics[width=0.35\textwidth]{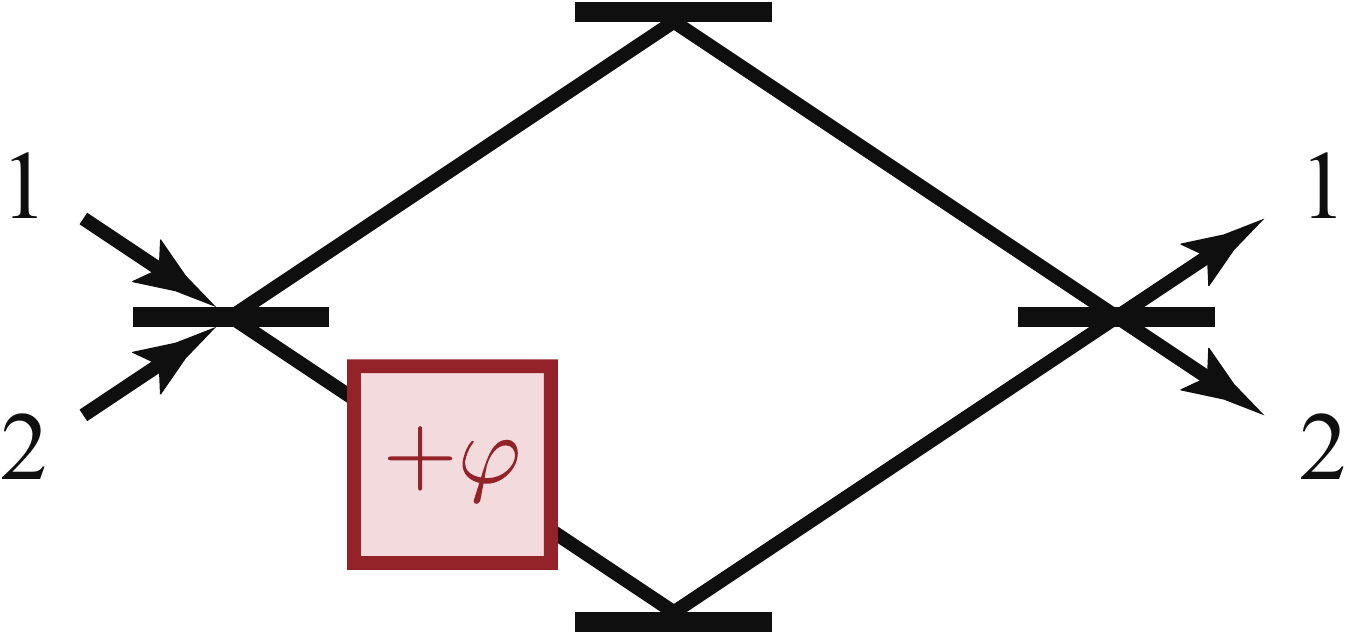}}
&
\makebox(155,72){\includegraphics[width=0.35\textwidth]{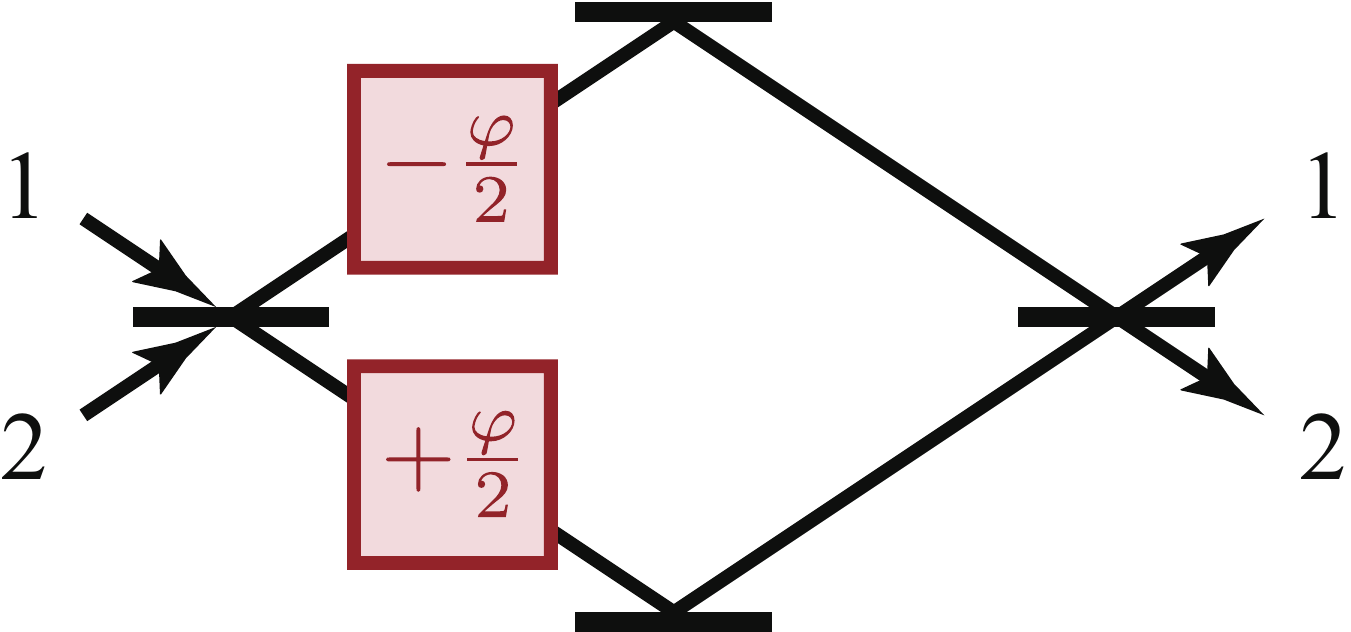}}
\\[4truemm]
(c) two-mode mixing
&
(d) three-mode mixing
\\[2truemm]
\makebox(155,63){\includegraphics[height=0.15\textwidth]{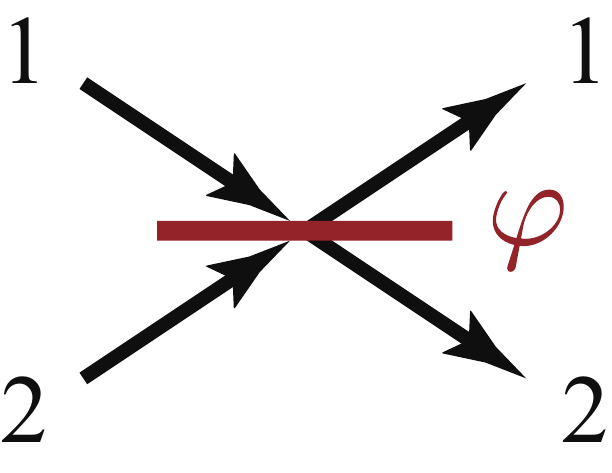}}
&
\makebox(155,63){\includegraphics[height=0.15\textwidth]{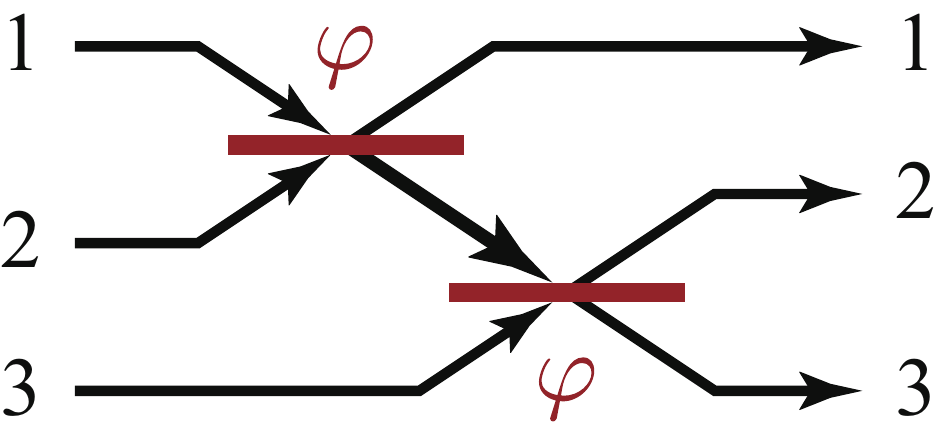}}
\end{tabular}
\end{indented}
\caption{(a)--(b) Two different arrangements of the MZ interferometer. (c) Two-mode mixing circuit. (d) Three-mode mixing circuit.}
\label{fig:MZ}
\end{figure}

\begin{table}[t]
\caption{\label{tab:Examples}%
The optimal Gaussian input state
$
\ket{\psi_\mathrm{opt}}
$ 
and the maximal QFI
$
F(\varphi|\ket{\psi_\mathrm{opt}})
$ 
for the estimation of the parameter $\varphi$ in each of the circuits shown in Fig.~\ref{fig:MZ}.
}
\begin{indented}
\item[]
\begin{tabular}{@{}ccc@{}}
\br
&
$
\ket{\psi_\mathrm{opt}}
$&
$
F(\varphi|\ket{\psi_\mathrm{opt}})
$
\\\hline
MZ interferometer I&
$
\hat{U}_{12}^\dag(\case{\pi}{4})\hat{S}_1(r_0)\ket{0}
$&
$
8\overline{N}(\overline{N}+1)
$
\\
MZ interferometer II&
$
\hat{U}_{12}^\dag(\case{\pi}{4})
\hat{S}_1(r_0)\ket{0}
$&
$
2\overline{N}(\overline{N}+1)
$
\\
two-mode mixing&
$
e^{\frac{\pi i}{4}(\hat{a}_1^\dag\hat{a}_1-\hat{a}_2^\dag\hat{a}_2)}
\hat{U}_{12}^\dag(\case{\pi}{4})\hat{S}_1(r_0)\ket{0}$&
$
8\overline{N}(\overline{N}+1)
$
\\
three-mode mixing&
$
\hat{U}_{12}^\dag(\varphi)
e^{\frac{\pi i}{2}\hat{a}_2^\dag\hat{a}_2}
\hat{U}_{31}(\case{\pi}{4})
\hat{U}_{12}(\case{\pi}{4})
\hat{S}_1(r_0)\ket{0}$&
$
16\overline{N}(\overline{N}+1)
$
\\
\br
\end{tabular}
\end{indented}
\end{table}

\section{Simple Examples}
\label{sec:Examples}
Let us look at a few simple examples, i.e.~the two- and three-mode circuits shown in Fig.~\ref{fig:MZ}, to see in particular how the unitary $\hat{V}_\varphi$ involved in the optimal input Gaussian state~(\ref{eqn:OptPure}) looks like.
The optimal input Gaussian states $
\ket{\psi_\mathrm{opt}}
$ and the maximal QFIs $
F(\varphi|\ket{\psi_\mathrm{opt}})
$ for those examples are summarized in Table \ref{tab:Examples}\@.

\subsection{Mach-Zehnder Interferometer I}
We first consider the Mach-Zehnder (MZ) interferometer in Fig.~\ref{fig:MZ}(a).
Our target is the phase shift $\varphi$ in one of the two arms of the interferometer.
The state of the probe photons going through this MZ interferometer is transformed by the unitary transformation
\begin{equation}
\hat{U}_\varphi
=\hat{U}_{12}^\dag(\case{\pi}{4})
e^{-i\varphi\hat{a}_1^\dag\hat{a}_1}
\hat{U}_{12}(\case{\pi}{4}),
\end{equation}
where
\begin{equation}
\hat{U}_{mn}(\theta)
=e^{\theta(
\hat{a}_n^\dag\hat{a}_m
-\hat{a}_m^\dag\hat{a}_n
)}
\end{equation}
describes a beam splitter for modes $m$ and $n$, which acts on the canonical operators as
\begin{eqnarray}
\left(
\begin{array}{c}
\medskip
\hat{U}_{mn}^\dag(\theta)
\hat{a}_m
\hat{U}_{mn}(\theta)
\\
\hat{U}_{mn}^\dag(\theta)
\hat{a}_n
\hat{U}_{mn}(\theta)
\end{array}
\right)
=
\left(
\begin{array}{cc}
\medskip
\cos\theta & -\sin\theta \\
\sin\theta & \cos\theta
\end{array}
\right)
\left(
\begin{array}{c}
\medskip
\hat{a}_m
\\
\hat{a}_n
\end{array}
\right)
=U_{mn}(\theta)
\left(
\begin{array}{c}
\medskip
\hat{a}_m
\\
\hat{a}_n
\end{array}
\right)
,
\nonumber\\
\end{eqnarray}
with $\theta$ characterizing its transmissivity.
In particular, $\hat{U}_{mn}(\case{\pi}{4})$ describes a balanced beam splitter.
The generator of this two-mode circuit reads
\begin{equation}
\hat{G}_\varphi
=i\hat{U}_\varphi^\dag
\frac{\rmd\hat{U}_\varphi}{\rmd\varphi}
=\hat{U}_{12}^\dag(\case{\pi}{4})
\hat{a}_1^\dag\hat{a}_1
\hat{U}_{12}(\case{\pi}{4}).
\end{equation}
The unitary matrix $U_\varphi$ related to the unitary transformation $\hat{U}_\varphi$ through~(\ref{eqn:UnitaryRotation}) is given by
\begin{equation}
U_\varphi
=
U_{12}^\dag(\case{\pi}{4})
\left(
\begin{array}{cc}
\medskip
e^{-i\varphi} & 0 \\
0 & 1
\end{array}
\right)
U_{12}(\case{\pi}{4}),
\end{equation}
and its generator reads
\begin{equation}
g_\varphi
=iU_\varphi^\dag\frac{\rmd U_\varphi}{\rmd\varphi}
=
U_{12}^\dag(\case{\pi}{4})
\left(
\begin{array}{cc}
\medskip
1 & 0 \\
0 & 0
\end{array}
\right)
U_{12}(\case{\pi}{4}).
\label{eqn:gMZ1}
\end{equation}
We thus have
\begin{equation}
\|g_\varphi\|=1.
\end{equation}
The unitary operator $\hat{V}_\varphi$ corresponding to the unitary matrix diagonalizing $g_\varphi$ in (\ref{eqn:gMZ1}) [compare it with (\ref{eqn:Hamiltonian11})] is
\begin{equation}
\hat{V}_\varphi=\hat{U}_{12}^\dag(\case{\pi}{4}).
\end{equation}
Therefore, the optimal Gaussian input state~(\ref{eqn:OptPure}) for this MZ interferometer is given by
\begin{equation}
\ket{\psi_\mathrm{opt}}
=\hat{U}_{12}^\dag(\case{\pi}{4})\hat{S}_1(r_0)\ket{0},
\label{eqn:PsiOptMZ1}
\end{equation}
with the squeezing parameter $r_0$ given in~(\ref{eqn:r0}).
By this choice, the QFI reaches the upper bound in~(\ref{eqn:OptQFISeq}), yielding
\begin{equation}
F(\varphi|\ket{\psi_\mathrm{opt}})
=8\overline{N}(\overline{N}+1).
\label{eqn:QFIOptMZ1}
\end{equation}

Notice that, in this case, the optimal input state $\ket{\psi_\mathrm{opt}}$ in~(\ref{eqn:PsiOptMZ1}) is \textit{independent} of the target parameter $\varphi$.
Note also that the same expression as~(\ref{eqn:QFIOptMZ1}) is found e.g.~in Refs.~\cite{PhysRevA.73.033821,PhysRevA.79.033834,PhysRevA.88.040102,PhysRevA.91.013808,PhysRevA.94.062313}, but it is found there as the optimal QFI for the estimation of the \textit{single-mode} phase shift with a Gaussian probe. Here, (\ref{eqn:QFIOptMZ1}) is presented as the optimal QFI for the \textit{two-mode} circuit in Fig.~\ref{fig:MZ}(a).

The unitary transformation $\hat{U}_{12}^\dag(\case{\pi}{4})$ in the optimal input state~(\ref{eqn:PsiOptMZ1}) ``unfolds'' the first beam splitter $\hat{U}_{12}(\case{\pi}{4})$ of the MZ interferometer.
Thus, the best strategy effectively consists in sending the single-mode squeezed vacuum $\ket{r_0}=\hat{S}_1(r_0)\ket{0}$ directly to the phase shifter without the first beam splitter $\hat{U}_{12}(\case{\pi}{4})$.
The second beam splitter $\hat{U}_{12}^\dag(\case{\pi}{4})$ of the MZ interferometer is also unfolded by $\hat{U}_{12}(\case{\pi}{4})$ performed in the optimal measurements [see~(\ref{eqn:OptMeas1}) and~(\ref{eqn:HomodynePOVM}), where $\hat{U}_\varphi$ contains $\hat{U}_{12}^\dag(\case{\pi}{4})$, whose Hermitian conjugate $\hat{U}_{12}(\case{\pi}{4})$ in $\hat{U}_\varphi^\dag$ acts on the output probe state first in the measurement process, cancelling the second beam splitter $\hat{U}_{12}^\dag(\case{\pi}{4})$].

\subsection{Mach-Zehnder Interferometer II}
Let us look at the MZ interferometer in the slightly different configuration shown in Fig.~\ref{fig:MZ}(b).
This setup induces the unitary transformation
\begin{equation}
\hat{U}_\varphi
=\hat{U}_{12}^\dag(\case{\pi}{4})
e^{-i\frac{\varphi}{2}(\hat{a}_1^\dag\hat{a}_1
-\hat{a}_2^\dag\hat{a}_2)}
\hat{U}_{12}(\case{\pi}{4}),
\label{eqn:UMZ2}
\end{equation}
and its generator is given by
\begin{equation}
\hat{G}_\varphi
=\frac{1}{2}
\hat{U}_{12}^\dag(\case{\pi}{4})
(
\hat{a}_1^\dag\hat{a}_1
-\hat{a}_2^\dag\hat{a}_2
)
\hat{U}_{12}(\case{\pi}{4}).
\label{eqn:GeneMZ2}
\end{equation}
The unitary matrix $U_\varphi$ corresponding to the unitary operator $\hat{U}_\varphi$ in (\ref{eqn:UMZ2}) is given by
\begin{equation}
U_\varphi
=
U_{12}^\dag(\case{\pi}{4})
\left(
\begin{array}{cc}
\medskip
e^{-i\varphi/2} & 0 \\
0 & e^{i\varphi/2}
\end{array}
\right)
U_{12}(\case{\pi}{4}),
\end{equation}
while the Hermitian matrix $g_\varphi$ corresponding to the generator $\hat{G}_\varphi$ in (\ref{eqn:GeneMZ2}) reads
\begin{equation}
g_\varphi
=
U_{12}^\dag(\case{\pi}{4})
\left(
\begin{array}{cc}
\medskip
1/2 & 0 \\
0 & -1/2
\end{array}
\right)
U_{12}(\case{\pi}{4}).
\end{equation}
We thus have
\begin{equation}
\|g_\varphi\|=\frac{1}{2}.
\end{equation}
The optimal Gaussian input state for this MZ interferometer is the same as the one given in~(\ref{eqn:PsiOptMZ1}), while the maximal QFI achievable by the optimal input state is
\begin{equation}
F(\varphi|\ket{\psi_\mathrm{opt}})
=2\overline{N}(\overline{N}+1).
\label{eqn:QFIOptMZ2}
\end{equation}

This QFI is lower than the previous one in~(\ref{eqn:QFIOptMZ1}) for the other MZ interferometer, even though the relative phases $\varphi$ to be estimated in the two MZ interferometers are the same. This is because injecting all the resources to one of the two arms of the interferometer is optimal if we stick to Gaussian probes, and only one of the two phase shifters in Fig.~\ref{fig:MZ}(b) is probed. 
It would be worth noticing that our estimation problem implicitly assumes the presence of an external phase reference. 
Without the reference beam, the two MZ interferometers in Figs.~\ref{fig:MZ}(a) and (b) are equivalent, since only the relative phase between the two arms matters in such a case.
See the discussion in Ref.~\cite{PhysRevA.85.011801}.

\subsection{Two-Mode Mixing}
Let us look at another two-mode example: the estimation of the parameter $\varphi$ characterizing the transmissivity of the beam splitter represented by the unitary transformation
\begin{equation}
\hat{U}_\varphi
=\hat{U}_{12}(\varphi).
\label{eqn:TwoModeMixing}
\end{equation}
See Fig.~\ref{fig:MZ}(c).
Its generator reads
\begin{equation}
\hat{G}_\varphi
=i(
\hat{a}_2^\dag\hat{a}_1
-\hat{a}_1^\dag\hat{a}_2
),
\end{equation}
which can be rewritten as
\begin{equation}
\hat{G}_\varphi
=
e^{\frac{\pi i}{4}(\hat{a}_1^\dag\hat{a}_1-\hat{a}_2^\dag\hat{a}_2)}
\hat{U}_{12}^\dag(\case{\pi}{4})
(
\hat{a}_1^\dag\hat{a}_1
-
\hat{a}_2^\dag\hat{a}_2
)
\hat{U}_{12}(\case{\pi}{4})
e^{-\frac{\pi i}{4}(\hat{a}_1^\dag\hat{a}_1-\hat{a}_2^\dag\hat{a}_2)}.
\end{equation}
It is unitarily equivalent to the generator $\hat{G}_\varphi
$ in~(\ref{eqn:GeneMZ2}), apart from the numerical proportionality constant $1/2$.
We thus have
\begin{equation}
\|g_\varphi\|=1,
\end{equation}
and the maximal QFI is given by
\begin{equation}
F(\varphi|\ket{\psi_\mathrm{opt}})
=8\overline{N}(\overline{N}+1).
\label{eqn:QFIOptBS}
\end{equation}
This is reached by the input state
\begin{equation}
\ket{\psi_\mathrm{opt}}
=e^{\frac{\pi i}{4}(\hat{a}_1^\dag\hat{a}_1-\hat{a}_2^\dag\hat{a}_2)}
\hat{U}_{12}^\dag(\case{\pi}{4})\hat{S}_1(r_0)\ket{0},
\label{eqn:PsiOptBS}
\end{equation}
with the squeezing parameter $r_0$ given in~(\ref{eqn:r0}).
This optimal state is again independent of the target parameter $\varphi$.

The same estimation problem, i.e.~the estimation of $\varphi$ in the two-mode mixing channel (\ref{eqn:TwoModeMixing}), is studied in Ref.~\cite{PhysRevA.94.062313}, but the maximal QFI (\ref{eqn:QFIOptBS}) and the optimal Gaussian input state (\ref{eqn:PsiOptBS}) are not identified there.

\subsection{Three-Mode Mixing}
Let us also look at a three-mode example.
We consider the circuit shown in Fig.~\ref{fig:MZ}(d), composed of two beam splitters of the same transmissivity characterized by the parameter $\varphi$.
Our problem is to estimate the single parameter $\varphi$ in the three-mode mixing circuit represented by the unitary transformation
\begin{equation}
\hat{U}_\varphi
=	
\hat{U}_{23}(\varphi)
\hat{U}_{12}(\varphi).
\end{equation}
Its generator reads
\begin{eqnarray}
\hat{G}_\varphi
&=&
i
\hat{U}_{12}^\dag(\varphi)
(
\hat{a}_3^\dag\hat{a}_2
-\hat{a}_2^\dag\hat{a}_3
+
\hat{a}_2^\dag\hat{a}_1
-\hat{a}_1^\dag\hat{a}_2
)
\hat{U}_{12}(\varphi)
\nonumber\\
&=&\sqrt{2}\,
\hat{V}_\varphi
(
\hat{a}_1^\dag\hat{a}_1
-
\hat{a}_2^\dag\hat{a}_2
)
\hat{V}_\varphi^\dag,
\end{eqnarray}
with
\begin{equation}
\hat{V}_\varphi
=	
\hat{U}_{12}^\dag(\varphi)
e^{\frac{\pi i}{2}\hat{a}_2^\dag\hat{a}_2}
\hat{U}_{31}(\case{\pi}{4})
\hat{U}_{12}(\case{\pi}{4}).
\end{equation}
We have
\begin{equation}
\|g_\varphi\|=\sqrt{2},
\end{equation}
and the maximal QFI is given by
\begin{equation}
F(\varphi|\ket{\psi_\mathrm{opt}})
=16\overline{N}(\overline{N}+1).
\label{eqn:QFIOptThreeMode}
\end{equation}
This is reached by the input state
\begin{equation}
\ket{\psi_\mathrm{opt}}
=
\hat{U}_{12}^\dag(\varphi)
e^{\frac{\pi i}{2}\hat{a}_2^\dag\hat{a}_2}
\hat{U}_{31}(\case{\pi}{4})
\hat{U}_{12}(\case{\pi}{4})
\hat{S}_1(r_0)\ket{0},
\label{eqn:PsiOptThreeMode}
\end{equation}
with the squeezing parameter $r_0$ given in~(\ref{eqn:r0}).
In this case, the optimal input state depends on the target parameter $\varphi$.

If our guess $\varphi'$ is not precise and does not match the true value $\varphi$, the input state (\ref{eqn:OptPure}) and the measurement, e.g.~(\ref{eqn:POVM1}) or (\ref{eqn:HomodynePOVM}), prepared and performed with the guessed value $\varphi'$ in place of $\varphi$ (see e.g.~the circuit in Fig.~\ref{fig:MeasCircuit}) are not optimal, and the FI for such a nonoptimal probing deviates from the maximal QFI in~(\ref{eqn:OptQFISeq}).
Since we assume that the functional dependence of $\hat{U}_\varphi$ upon $\varphi$ is smooth, the FI is a smooth function of $\varphi'$, and therefore, the deviation of FI from the maximal QFI is only quadratic around the optimal point $\varphi'=\varphi$.
In this sense, the FI is robust to a small error in the guess of $\varphi$.

\section{Sequential Strategy}
\label{sec:Sequential} 
If we are allowed to use multiple (identical) target circuits $\hat{U}_\varphi$ at the same time, we could do better.
Suppose that we are given $L$ identical $M$-mode passive linear circuits $\hat{U}_\varphi$.
A paradigmatic scheme for the quantum metrology is the parallel scheme in Fig.~\ref{fig:SequentialSchemes}(a) with an entangled input $\hat{\rho}$~\cite{ref:QuantumMetrologyVittorio,ref:MetrologyNaturePhoto}.
The result in Sec.~\ref {sec:Optimization} suggests, however, that, if we stick to Gaussian inputs, this parallel setup does not help improve the maximal QFI found in~(\ref{eqn:OptQFISeq}), since the best strategy is to inject all the resources into a single mode of the overall $LM$-mode passive linear circuit in Fig.~\ref{fig:SequentialSchemes}(a): only one of the $L$ circuits is probed with the others irrelevant.
See~(\ref{eqn:OptPure}).
On the other hand, if we are allowed to perform some operations $\{\hat{U}_1,\ldots,\hat{U}_{L-1}\}$ between the target gates $\hat{U}_\varphi$ with ancilla modes introduced as in Fig.~\ref{fig:SequentialSchemes}(c), we can hope to do better.
Let us restrict ourselves to passive linear controls $\{\hat{U}_1,\ldots,\hat{U}_{L-1}\}$, and seek for the optimal strategy with a Gaussian input $\hat{\rho}\in\mathcal{G}(K,\overline{N})$, where $K\ge LM$.

\begin{figure*}
\centering
\makebox(0,90)[tr]{\footnotesize(a)}%
\makebox(69,90)[t]{
\includegraphics[scale=0.26]{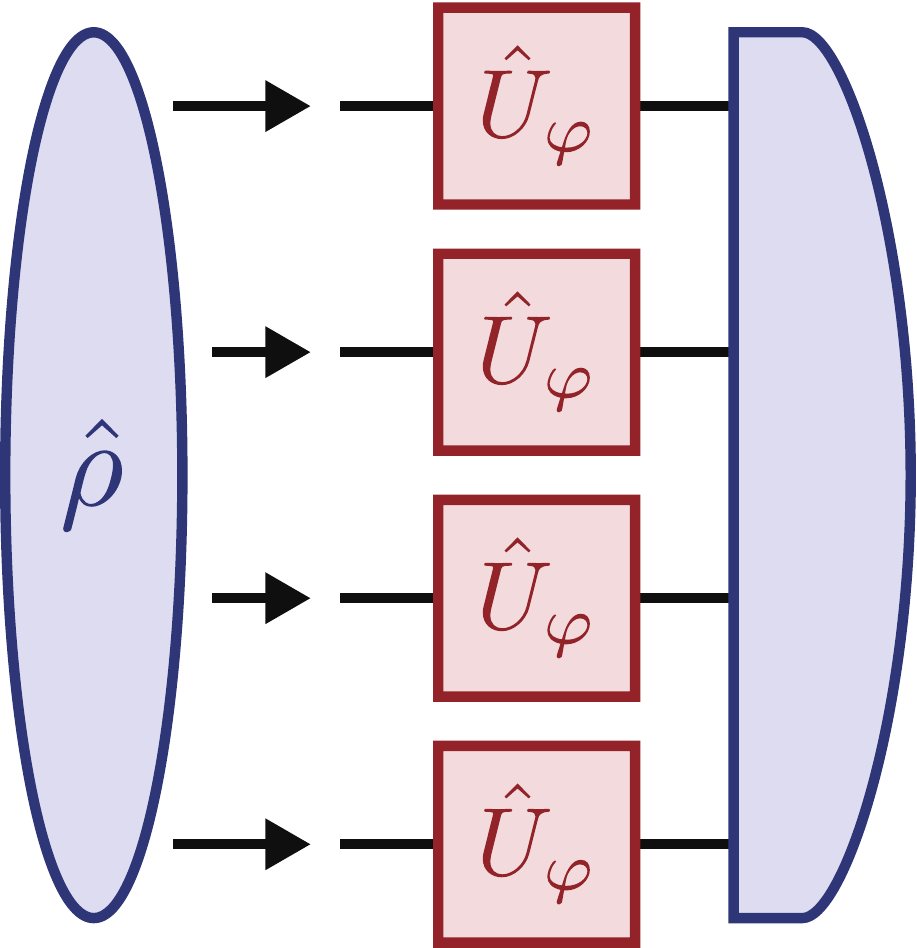}}
\makebox(27,90)[t]{\makebox(7,71){$=$}}%
\makebox(0,90)[tr]{\footnotesize(b)}%
\makebox(162,90)[t]{\includegraphics[scale=0.26]{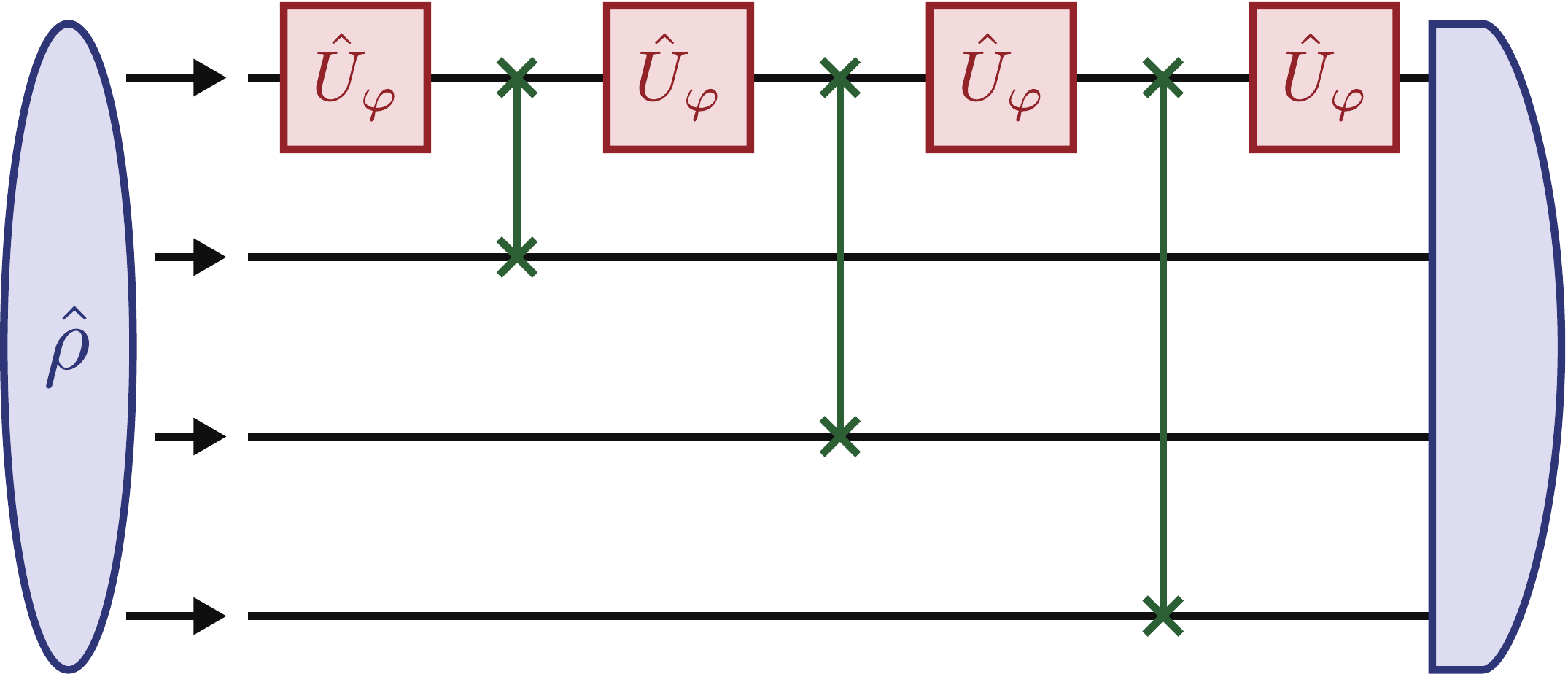}}%
\makebox(27,90)[t]{\makebox(7,71){$\subseteq$}}%
\makebox(0,90)[tr]{\footnotesize(c)}%
\makebox(161,90)[t]{\includegraphics[scale=0.26]{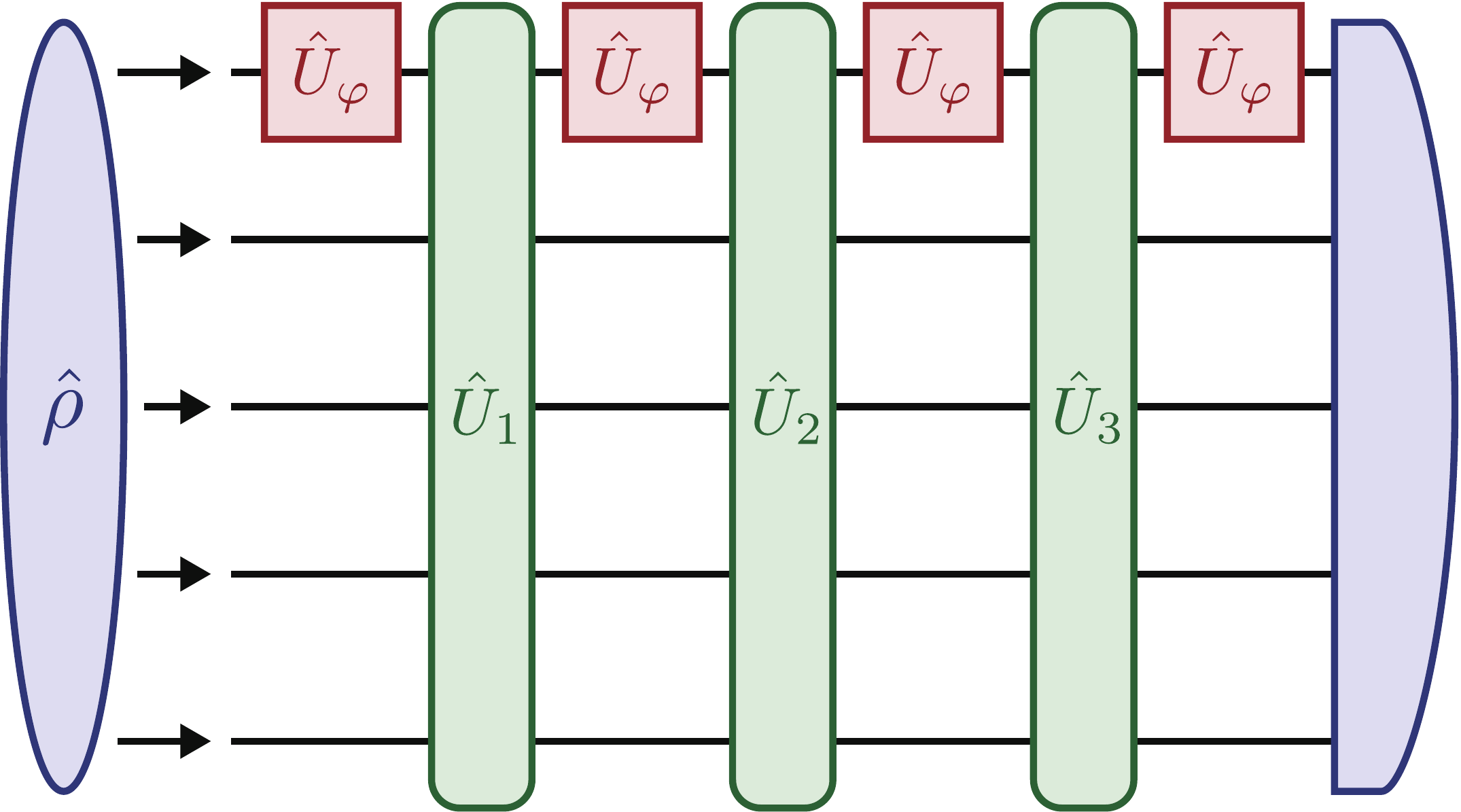}}
\caption{The parallel scheme in (a) is equivalent to the sequential scheme in (b) with the target circuits $\hat{U}_\varphi$ swapped by \textsc{swap} gates, which is a particular case of the sequential scheme in (c) with generic gates $\hat{U}_\ell$ entangling the main probes with additional ancillas.}
\label{fig:SequentialSchemes}
\end{figure*}
The circuit in Fig.~\ref{fig:SequentialSchemes}(c) is described by the unitary
\begin{equation}
\hat{\mathcal{U}}_\varphi
=\hat{U}_\varphi
\hat{U}_{L-1}
\hat{U}_\varphi
\cdots
\hat{U}_2
\hat{U}_\varphi
\hat{U}_1
\hat{U}_\varphi.
\label{eqn:SeqCircuit}
\end{equation}
Note that there are $K\,(\ge LM)$ modes in total in the overall circuit, and the unitary operators $\hat{U}_\varphi$ act only on the first $M$ modes, i.e.~$\hat{U}_\varphi\otimes\openone$.
By abuse of notation, $\hat{U}_\varphi\otimes\openone$ is simply denoted by $\hat{U}_\varphi$ in~(\ref{eqn:SeqCircuit}).
The overall circuit is a $K$-mode passive linear circuit, and the orthogonal matrix $\mathcal{R}_\varphi$ which rotates the quadrature operators $\hat{\bm{z}}$ in  phase space according to the transformation $\hat{\mathcal{U}}_\varphi$ is given by
\begin{equation}
\mathcal{R}_\varphi
=
W^\dag
\left(\begin{array}{@{}c|c@{}}
\,\,\,\mathcal{U}_\varphi & 0\,\, \\
\hline
\,\,\,0 & \mathcal{U}^*_\varphi\,\,
\end{array}\right)
W
\end{equation}
with
\begin{equation}
\mathcal{U}_\varphi
=U_\varphi
U_{L-1}
U_\varphi\cdots
U_2
U_\varphi
U_1
U_\varphi,
\end{equation}
where $U_\varphi$ and $U_\ell$ ($\ell=1,\ldots,L-1$) are $K\times K$ unitary matrices corresponding to $\hat{U}_\varphi\otimes\openone$ and $\hat{U}_\ell$, respectively.
The quantity relevant to the maximal QFI is the spectral norm of the generator of this orthogonal transformation $\mathcal{R}_\varphi$ [see~(\ref{eqn:OptQFISeq})], i.e.~the largest (in magnitude) eigenvalue of
\begin{equation}
\mathcal{G}_\varphi
=i\mathcal{U}_\varphi^\dag\frac{\rmd \mathcal{U}_\varphi}{\rmd\varphi}
=
\sum_{\ell=0}^{L-1}
U_\varphi^\dag
U_1^\dag\cdots
U_\varphi^\dag
U_\ell^\dag
g_\varphi U_\ell
U_\varphi\cdots
U_1U_\varphi,
\end{equation}
where
\begin{equation}
g_\varphi
=iU_\varphi^\dag\frac{\rmd U_\varphi}{\rmd\varphi}.
\label{eqn:Hamiltonian}
\end{equation}
The spectral norm of the generator $\mathcal{G}_\varphi$ is  bounded from above as
\begin{eqnarray}
\|\mathcal{G}_\varphi\|
&
=
\left\|
\sum_{\ell=0}^{L-1}
U_\varphi^\dag
U_1^\dag\cdots
U_\varphi^\dag
U_\ell^\dag
g_\varphi
U_\ell
U_\varphi\cdots
U_1U_\varphi
\right\|
\nonumber\\
&
\le
\sum_{\ell=0}^{L-1}
\|
U_\varphi^\dag
U_1^\dag
\cdots
U_\varphi^\dag
U_\ell^\dag
g_\varphi U_\ell U_\varphi\cdots U_1U_\varphi
\|
=L\|g_\varphi\|.
\vphantom{\sum_{\ell=0}^{L-1}}
\end{eqnarray}
This inequality is saturated if
\begin{equation}
[
g_\varphi,U_\ell U_\varphi
]=0\qquad
(\ell=1,\ldots,L-1).
\end{equation}
A sufficient and general solution is given by 
\begin{equation}
U_\ell=U_\varphi^\dag\qquad
(\ell=1,\ldots,L-1)
\label{eqn:OptControls}
\end{equation}
(cf.~\cite{ref:EstCtrlInversion}).
By this choice, the generator of the overall circuit $\hat{\mathcal{U}}_\varphi$ is reduced to $\mathcal{G}_\varphi=Lg_\varphi$, and the upper bound on the QFI by the sequential strategy with a Gaussian input $\hat{\rho}\in\mathcal{G}(K,\overline{N})$ is given by
\begin{equation}
\mathcal{F}(\varphi|\hat{\rho})
\le8L^2\|g_\varphi\|^2\overline{N}(\overline{N}+1).
\end{equation}
This upper bound is saturated by the input state
\begin{equation}
\ket{\Psi_\mathrm{opt}}
=\ket{\psi_\mathrm{opt}}\otimes\ket{0},
\label{eqn:OptPsiSeq}
\end{equation}
with $\ket{\psi_\mathrm{opt}}$ given in~(\ref{eqn:OptPure}) for the first $M$ modes while vacuum for the rest.

The results in~(\ref{eqn:OptControls}) and~(\ref {eqn:OptPsiSeq}) show that the ancilla modes are not necessary for the optimal strategy.
We note that in general the optimal controls~(\ref{eqn:OptControls}) and the optimal input state~(\ref {eqn:OptPsiSeq}) depend on the target parameter $\varphi$.

\section{Summary}
\label{sec:Conclusion}
We have clarified the universal bound~(\ref{eqn:OptQFISeq}) on the precision of the estimation (QFI) of a parameter embedded in a generic multimode passive (photon-number preserving) linear optical circuit by using Gaussian probes with a given average number of probe photons $\overline{N}$.
We have identified the input Gaussian state~(\ref{eqn:OptPure}) that yields the QFI saturating the bound~(\ref{eqn:OptQFISeq}): it is a single-mode squeezed vacuum in an appropriate basis.
We have also found measurements (POVMs)~(\ref{eqn:POVM1}) and~(\ref{eqn:HomodynePOVM}) by which FI reaches QFI\@. 
The best (sequential) strategy when we are given multiple identical target circuits and are allowed to apply passive linear controls in between with the help of an arbitrary number of ancilla modes has been revealed: no ancilla mode is actually needed for the best strategy.\footnote{There are works in the literature which discuss the
  unnecessity of mode entanglement~\cite
  {DemkowiczDobrzanski2015345,PhysRevA.90.033846,PhysRevA.91.013808,PhysRevA.93.033859,PhysRevA.94.033817,PhysRevA.94.042342,PhysRevA.70.032310,PhysRevLett.120.080501}.
  Note, however, that in those works the probe states are not restricted to
  Gaussian states and in addition just the achievability of the Heisenberg
  scaling (quadratic in $\protect \overline {N}$) is discussed. The chosen
  probe states are not necessarily the optimal ones, even though they actually
  yields QFIs scaling quadratically in $\protect \overline {N}$ (their
  coefficients are not necessarily the optimal). On the other hand, in the
  present work, we look at the optimal state which yields the maximal QFI\@.}

Even though the optimal input state~(\ref{eqn:OptPure}) and the optimal measurements~(\ref{eqn:POVM1}) and~(\ref{eqn:HomodynePOVM}), as well as the optimal controls~(\ref{eqn:OptControls}) in the sequential strategy, depend on the target parameter to be estimated in general and adaptive adjustments of the input, the measurement, and the controls would be required to achieve the precision bound in practice, the above result shows that the bound is sharp and covers various specific setups composed of phase shifters and beam splitters, including the standard MZ interferometer, providing the universal bound that cannot be beaten by any Gaussian inputs and any passive controls.

The present work has focussed on passive linear circuits.
Bounds on more general Gaussian metrology, for general Gaussian channels including amplitude-damping channels and channels involving squeezing, etc., have not been thoroughly understood yet, beyond analyses on specific setups.
Entanglement with ancilla modes would be useful for such generic Gaussian metrology \cite{PhysRevLett.101.253601} and it would be interesting to explore.

\section*{Acknowledgments}
KY thanks Koji Matsuoka for the discussions during his master's thesis study \cite{ref:MatsuokaMS}, in which the bound (\ref{eqn:OptQFISeq}) and the optimal input state (\ref{eqn:OptPure}) were found for some restricted setups.
This work was supported by the Top Global University Project from the Ministry of Education, Culture, Sports, Science and Technology (MEXT), Japan.
KY was supported by the Grant-in-Aid for Scientific Research (C) (No.~18K03470) from the Japan Society for the Promotion of Science (JSPS) and by the Waseda University Grant for Special Research Projects (No.~2018K-262).
PF was supported by INFN through the project ``QUANTUM,'' and by the Italian National Group of Mathematical Physics (GNFM-INdAM).

\appendix
\section{Gaussian States and Operations}\label{sec:Setup}
In order to introduce a proper definition of the Gaussian set $\mathcal{G}(M,\overline{N})$, we find it useful to 
  introduce the quadrature operators $\hat{x}_m$ and $\hat{y}_m$ for each of the $M$ modes, 
\begin{equation}
\left\{\begin{array}{l}
\medskip
\displaystyle
\hat{x}_m=\frac{\hat{a}_m+\hat{a}_m^\dag}{\sqrt{2}}\\
\displaystyle
\hat{y}_m=\frac{\hat{a}_m-\hat{a}_m^\dag }{\sqrt{2}i}
\end{array}\right.
\qquad(m=1,\ldots ,M).
\label{eqn:Quadratures}
\end{equation}
Aligning these operators as a column vector
\begin{equation}
\hat{\bm{z}}=
\left(\begin{array}{@{}c@{}}
\,\,\hat{\bm{x}}\,\, \\
\hline
\,\,\hat{\bm{y}}\,\,
\end{array}\right)
=
\left(\begin{array}{@{}c@{}}
\,\,\hat{x}_1\,\,\\
\vdots \\
\,\,\hat{x}_M\,\,\\
\hline
\,\,\hat{y}_1\,\,\\
\vdots \\
\,\,\hat{y}_M\,\,
\end{array}\right), \label{QUADR} 
\end{equation}
the above relation~(\ref{eqn:Quadratures}) can be expressed as
\begin{equation}
\left(\begin{array}{@{}c@{}}
\,\,\hat{\bm{a}}\,\,\\
\hline
\,\,\hat{\bm{a}}^\dag\,\,
\end{array}\right)
=W\left(\begin{array}{@{}c@{}}
\,\,\hat{\bm{x}}\,\,\\
\hline
\,\,\hat{\bm{y}}\,\,
\end{array}\right)
\end{equation}
with a $2M\times2M$ unitary matrix
\begin{equation}
W=
\frac{1}{\sqrt{2}}
\left(\begin{array}{@{}c|r@{}}
\,\,\,\openone&i\openone\,\,\\
\hline
\,\,\,\openone&-i\openone\,\,
\end{array}\right).
\end{equation}
The canonical commutation relations~(\ref{CANO}) can then be expressed in the compact form 
\begin{equation}
[\hat{z}_m,\hat{z}_n]=i J_{mn} \qquad(m,n=1,\ldots,2M) ,
\end{equation}
with $J$ being the $2M\times 2M$ real matrix 
\begin{equation}
J=
\left(\begin{array}{@{}c|c@{}}
\,\,0&\openone\,\,\\
\hline
\,\,-\openone & 0\,\,
\end{array}\right). \label{DEFJJJ} 
\end{equation}

\subsection{Gaussian States} 
A Gaussian state $\hat{\rho}$ 
is fully characterized by its covariance matrix $\Gamma$ and its displacement $\bm{d}$, defined by
\begin{eqnarray}
\Gamma _{mn}=\frac{1}{2}\langle\{\hat{z}_m,\hat{z}_n\}\rangle
- \langle\hat{z}_m\rangle\langle\hat{z}_n\rangle,\quad
d_m=\langle\hat{z}_m\rangle\quad
(m,n=1,\ldots,2M),
\label{eqn:GammaD}
\end{eqnarray}
where $\langle{}\cdots{}\rangle$ denotes the expectation value on $\hat{\rho}$. 
In particular, its characteristic function reads as
\begin{equation}
\chi(\bm{\eta})
=\langle
e^{i\bm{\eta}\cdot\hat{\bm{z}}}
\rangle
=e^{-\frac{1}{2}\bm{\eta}^T\Gamma\bm{\eta}+i\bm{\eta}\cdot\bm{d}}.
\label{eqn:CharFunc}
\end{equation}
Furthermore, $\hat{\rho}$ is an element of $\mathcal{G}(M,\overline{N})$ when its mean photon number is equal to $\overline{N}$, i.e. 
\begin{equation}
\langle\hat{N}\rangle
=
\frac{1}{2}\left[
\Tr\!\left(
\Gamma-\frac{1}{2}
\right)
+\bm{d}^2
\right]
= \overline{N},
\label{eqn:constraint111}\end{equation} 
where the number operator $\hat{N}$ is defined in~(\ref{eqn:N}).
The covariance matrix $\Gamma$ is real, symmetric, and positive-definite, and hence, according to Williamson's theorem it admits the canonical decomposition~\cite{ref:GaussianQI,ref:ContVarQI-Adesso}
\begin{equation}
\Gamma =RQR'\Sigma R'^{T}QR^T,
\label{eqn:SymplecticDecomp}
\end{equation}
where
\begin{eqnarray}
\Sigma =
\left(\begin{array}{@{}c|c@{}}
\,\,\,\sigma & 0\,\, \\
\hline
\,\,\,0 & \sigma\,\, 
\end{array}\right),
\qquad
&Q=
\left(\begin{array}{@{}c|c@{}}
\,\,\,e^{r} & 0\,\, \\
\hline
\,\,\,0 & e^{-r}\,\,
\end{array}\right),
\label{eqn:SigmaQ}
\\
R
=
W^\dag
\left(\begin{array}{@{}c|c@{}}
\,\,\,U & 0\,\, \\
\hline
\,\,\,0 & U^*\,\,
\end{array}\right)
W,
\qquad
&R'
=
W^\dag
\left(\begin{array}{@{}c|c@{}}
\,\,\,U' & 0\,\, \\
\hline
\,\,\,0 & U'^*\,\,
\end{array}\right)
W,
\label{eqn:R}
\end{eqnarray}
with $M\times M$ diagonal submatrices
\begin{equation}
\sigma =
\left(\begin{array}{ccc}
\sigma_1&&\\[-1truemm]
&\ddots&\\
&&\sigma_M
\end{array}\right),\qquad
r=
\left(\begin{array}{ccc}
r_1&&\\[-1truemm]
&\ddots&\\
&&r_M
\end{array}\right),
\label{eqn:SigmaSqueezing}
\end{equation}
and $M\times M$ unitary submatrices $U$ and $U'$.\footnote{Note that $U^*$ is not the Hermitian conjugate $U^\dag$ of the $M\times M$ matrix $U$, but is obtained by taking the complex conjugate of each matrix element of $U$. In other words, it is $U^*=(U^\dag)^T=(U^T)^\dag$, with $T$ denoting the matrix transpose. This $U^*$ is necessary in the structure of $R$ in (\ref{eqn:R}), for the symplectic character of $R$.} 
The $2M\times2M$ matrices $R$ and $R'$ are real orthogonal matrices, and we have $R^T=R^\dag=R^{-1}$ and $R'^T=R'^\dag=R'^{-1}$.
The parameters $\{\sigma_1,\ldots,\sigma_M\}$ are the symplectic eigenvalues of $\Gamma$, which control the purity $p(\hat{\rho})$ of the Gaussian state $\hat{\rho}$ through~\cite{ref:ContVarQI-Adesso}
\begin{equation}
p(\hat{\rho})
=\Tr\hat{\rho}^2
=\frac{1}{\sqrt{\det(2\Gamma)}}
=\prod_{m=1}^M\frac{1}{2\sigma_m},
\end{equation}
while $\{r_1,\ldots,r_M\}$ are the squeezing parameters.
The symplectic eigenvalues are  bounded from below by $\sigma_m\ge1/2$ ($m=1,\ldots,M$) due to the uncertainty principle~\cite{ref:GaussianQI,ref:ContVarQI-Adesso}.
The Gaussian state $\hat{\rho}$ is pure, $p(\hat{\rho})=1$, if and only if all the symplectic eigenvalues saturate the lower bounds $\sigma_m=1/2$ ($m=1,\ldots,M$).
Without loss of generality, we assume that 
\begin{equation}
\sigma_1\ge \sigma_2\ge\cdots\ge \sigma_M\ge\frac{1}{2},\qquad
r_1\ge r_2\ge\cdots\ge r_M\ge0.
\label{eqn:DescendingSigmaR}
\end{equation}
This reordering can always be done by arranging properly $R$ and $R'$.
The matrices $R$ and $R'$ are symplectic and orthogonal, characterized by the structure~(\ref{eqn:R}) with the unitary matrices $U$ and $U'$.
The squeezing matrix $Q$ is also symplectic.
The symplectic character of these matrices is characterized by
\begin{equation}
R^TJR=J,\qquad
R'^TJR'=J,\qquad
Q^TJQ=J.
\label{eqn:Symplectic}
\end{equation}

\subsection{$M$-Mode Passive Gaussian Unitary}
\label{app:PassiveLinear}
Our target circuit $\hat{U}_\varphi$ is a generic $M$-mode passive Gaussian unitary,  whose action is characterized by the $M\times M$ unitary matrix $U_\varphi$ introduced in~(\ref{eqn:UnitaryRotation}).
In terms of the quadrature operators $\hat{z}_m$, it is rephrased as
\begin{equation}
\hat{U}_\varphi^\dag\hat{z}_m\hat{U}_\varphi=\sum_{n=1}^{2M}(R_\varphi)_{mn}\hat{z}_n
\qquad(m=1,\ldots,2M),
\label{eqn:Rphi}
\end{equation}
or simply written as $\hat{U}_\varphi^\dag\hat{\bm{z}}\hat{U}_\varphi=R_\varphi\hat{\bm{z}}$, with $R_\varphi$ being the $2M \times 2M$ orthogonal matrix defined by
\begin{equation} \label{DEFERRE}
R_\varphi
=
W^\dag
\left(\begin{array}{@{}c|c@{}}
\,\,\,U_\varphi & 0\,\, \\
\hline
\,\,\,0 & U^*_\varphi\,\,
\end{array}\right)
W.
\end{equation}
As is clear from this structure, the matrix $R_\varphi$ is symplectic and orthogonal, and the passive linear transformation $\hat{U}_\varphi$ is a rotation on the phase space.

By construction the transformation  $\hat{U}_\varphi$ maps the set  $\mathcal{G}(M,\overline{N})$  into itself. 
In particular, given $\hat{\rho}\in \mathcal{G}(M,\overline{N})$,  the covariance matrix $\Gamma_\varphi$ and the displacement $\bm{d}_\varphi$ of the associated Gaussian output state $\hat{\rho}_\varphi$ in~(\ref{OUTPUT})  are obtained by rotating the covariance matrix $\Gamma$ and the displacement $\bm{d}$ of the input state $\hat{\rho}$ as
\begin{equation}
\Gamma_\varphi=R_\varphi\Gamma R_\varphi^T
,\qquad
\bm{d}_\varphi=R_\varphi\bm{d}.
\label{eqn:RotGaussApp}
\end{equation}
Note that they still fulfill the constraint~(\ref{eqn:constraint111}) due to the fact that $R_\varphi$ is orthogonal.

An important role on our problem is played by   the generator of the transformation $\hat{U}_\varphi$, i.e.~by the operator 
\begin{equation}
\hat{G}_\varphi=i\hat{U}_\varphi^\dag\frac{\rmd \hat{U}_\varphi}{\rmd \varphi},
\label{eqn:GhatApp}
\end{equation}
whose equivalent on the phase space reads 
\begin{equation}
G_\varphi
=iR_\varphi^{T}\frac{\rmd R_\varphi}{\rmd \varphi}
=
W^\dag
\left(\begin{array}{@{}c|c@{}}
\,\,\,g_\varphi & 0\,\,\\
\hline
\,\,\,0 & -g^*_\varphi\,\,
\end{array}\right)
W
\label{eqn:G}
\end{equation}
with 
\begin{equation}
g_\varphi
=iU_\varphi^\dag\frac{\rmd U_\varphi}{\rmd\varphi}.
\label{eqn:HamiltonianApp}
\end{equation}
This $g_\varphi$ is an $M\times M$ Hermitian matrix, that can be diagonalized by means of an $M\times M$ unitary matrix $V_\varphi$,
\begin{equation}
g_\varphi
=V_\varphi\varepsilon_\varphi V_\varphi^\dag,\qquad
\varepsilon_\varphi
=\left(\begin{array}{ccc}
\varepsilon_1&&\\[-1truemm]
&\ddots&\\
&&\varepsilon_M
\end{array}\right),
\label{eqn:Hamiltonian11}
\end{equation}
where, without loss of generality, the magnitudes of the eigenvalues $\varepsilon_m$ of $g_\varphi$ are ordered in decreasing order
\begin{equation}
|\varepsilon_1|\ge|\varepsilon_2|\ge\cdots\ge|\varepsilon_M|.
\label{eqn:DescendEpsilon}
\end{equation}
The generator $G_\varphi$ is accordingly diagonalized as
\begin{eqnarray}
G_\varphi
=P_\varphi 
W^\dag
\left(\begin{array}{@{}c|c@{}}
\,\,\,\varepsilon_\varphi & 0\,\,\\
\hline
\,\,\,0 & -\varepsilon_\varphi\,\,
\end{array}\right)
W
P_\varphi^T
=
P_\varphi
(\mathcal{E}_\varphi iJ)
P_\varphi^T,
\label{eqn:Hdiagonal}
\end{eqnarray}
where
\begin{equation}
P_\varphi=
W^\dag
\left(\begin{array}{@{}c|c@{}}
\,\,\,V_\varphi & 0\,\, \\
\hline
\,\,\,0 & V^*_\varphi\,\,
\end{array}\right)
W,
\qquad
\mathcal{E}_\varphi
=
\left(\begin{array}{@{}c|c@{}}
\,\,\,\varepsilon_\varphi & 0\,\,\\
\hline
\,\,\,0 & \varepsilon_\varphi\,\,
\end{array}\right).
\label{eqn:P}
\end{equation}

\section{Derivation of the Expression (\ref{eqn:F1}) for $F^{(1)}(\varphi|\hat{\rho})$}
\label{app:3.17}
Here, we show the derivation of the expression for $F^{(1)}(\varphi|\hat{\rho})$ in (\ref{eqn:F1}).
Notice first that $R$ in (\ref{eqn:R}) is a real matrix, and hence, 
\begin{equation}
R^T=R^\dag=R^{-1}
=W^\dag
\left(\begin{array}{@{}c|c@{}}
\,\,\,U^\dag & 0\,\, \\
\hline
\,\,\,0 & U^T\,\,
\end{array}\right)
W.
\label{eqn:Rdag}
\end{equation}
Inserting this into (\ref{eqn:GammaPure}), the covariance matrix $\Gamma$ of a pure Gaussian state is expressed as
\begin{eqnarray}
\Gamma
&=&
\frac{1}{2}
W^\dag
\left(\begin{array}{@{}c|c@{}}
\,\,\,U & 0\,\, \\
\hline
\,\,\,0 & U^*\,\,
\end{array}\right)
W
\left(\begin{array}{@{}c|c@{}}
\,\,\,e^{2r} & 0\,\, \\
\hline
\,\,\,0 & e^{-2r}\,\,
\end{array}\right)
W^\dag
\left(\begin{array}{@{}c|c@{}}
\,\,\,U^\dag & 0\,\, \\
\hline
\,\,\,0 & U^T\,\,
\end{array}\right)
W
\nonumber\\
&=&
\frac{1}{2}
W^\dag
\left(\begin{array}{@{}c|c@{}}
\,\,\,U\cosh2r\,U^\dag & U\sinh2r\,U^T\,\, \\
\hline
\,\,\,U^*\sinh2r\,U^\dag & U^*\cosh2r\,U^T\,\,
\end{array}\right)
W
,
\end{eqnarray}
and we have
\begin{eqnarray}
\Gamma^{-1}
&=&
2
W^\dag
\left(\begin{array}{@{}c|c@{}}
\,\,\,U\cosh2r\,U^\dag & -U\sinh2r\,U^T\,\, \\
\hline
\,\,\,-U^*\sinh2r\,U^\dag & U^*\cosh2r\,U^T\,\,
\end{array}\right)
W.
\end{eqnarray}
Then, inserting these and (\ref{eqn:G}) into the first line of (\ref{eqn:F1}), we get
\begin{eqnarray}
F^{(1)}(\varphi|\hat{\rho})
&=&\frac{1}{2}\Tr(
G_\varphi\Gamma^{-1}G_\varphi\Gamma-G_\varphi^2
)
\nonumber\\
&=&\frac{1}{2}\Tr[
(U^\dag g_\varphi U\cosh2r)^2
]
+\frac{1}{2}\Tr[
(\cosh2r\,U^Tg^*_\varphi U^*)^2
]
\nonumber\\
&&{}
+\Tr(
U^\dag g_\varphi U\sinh2r\,U^T g^*_\varphi U^*\sinh2r
)
-\frac{1}{2}
\Tr(g_\varphi^2)
-\frac{1}{2}
\Tr(g_\varphi^2)^*.
\nonumber\\
\end{eqnarray}
Since $g_\varphi$ is Hermitian and hence $g_\varphi^*=g_\varphi^T$, this is simplified to the expression in (\ref{eqn:F1}), noting $\Tr(A^T)=\Tr A$ for any matrix $A$.

\section{Some Useful Inequalities}
\label{app:Ineq}
\begin{lemma}
For Hermitian matrices $A$ and $B$,
\begin{equation}
\Tr[(AB)^2]\le\Tr(A^2B^2).
\label{eqn:Ineq1}
\end{equation}
The equality holds if and only if $[A,B]=0$.
\end{lemma}
\begin{proof} 
Since $i(AB-BA)$ is Hermitian,
\begin{eqnarray}
0
&\le\Tr\{[i(AB-BA)]^2\}
\nonumber\\
&
=-2\Tr[(AB)^2]+2\Tr(A^2B^2).
\end{eqnarray}
Therefore, the inequality~(\ref{eqn:Ineq1}) follows.
The equality holds if and only if $AB-BA=0$.
\end{proof}

\begin{lemma}
For Hermitian matrices $A$ and $B$,
\begin{equation}
\Tr(A^TB^TAB)\le\Tr(A^2 B^2).
\label{eqn:Ineq2}
\end{equation}
The equality holds if and only if $AB=(AB)^T$.
\end{lemma}
\begin{proof}
By noting the Hermitianity of $A$ and $B$,
\begin{eqnarray}
0&\le\Tr\{[AB-(AB)^T]^\dag[AB-(AB)^T]\}
\nonumber\\
&
=2\Tr(A^2B^2)-2\Tr(A^TB^TAB).
\end{eqnarray}
Therefore, the inequality~(\ref{eqn:Ineq2}) follows.
The equality holds if and only if $AB-(AB)^T=0$.
\end{proof}

\begin{lemma}
For Hermitian and positive semi-definite matrices $A$ and $B$,
\begin{equation}
\Tr(AB)\le\|A\|\Tr B,
\end{equation}
where $\|A\|$ is the spectral norm of $A$, given by its largest eigenvalue.
The equality holds if and only if the support of $B$ (i.e.\ the orthogonal complement of its kernel) is contained in the eigenspace of $A$ belonging to its largest eigenvalue.
\end{lemma}
\begin{proof}
Consider the spectral decomposition of the Hermitian and positive semi-definite matrix $A$,
\begin{equation}
A=\sum_n\lambda_n\bm{v}_n\bm{v}_n^\dag,\qquad
\lambda_n\ge0.
\end{equation}
Then, by noting the fact that $\bm{u}^\dag B\bm{u}\ge0$ for any vector $\bm{u}$,
\begin{equation}
\Tr(AB)
=\sum_n\lambda_n\bm{v}_n^\dag B\bm{v}_n
\le\lambda_\mathrm{max}\sum_n\bm{v}_n^\dag B\bm{v}_n
=\lambda_\mathrm{max}\Tr B,
\end{equation}
which proves the statement.
The equality holds if and only if $(\lambda_\mathrm{max}-\lambda_n)\bm{v}_n^\dag B\bm{v}_n=0$ for all $n$, i.e.~if and only if $\bm{v}_k^\dag B\bm{v}_k=0$ for all $\bm{v}_k$ belonging to the eigenvalues $\lambda_k$ of $A$ strictly smaller than $\lambda_\mathrm{max}$.
This, in turns, is equivalent to the condition that the support of $B$ is  in the eigenspace of $A$ belonging to its largest eigenvalue $\lambda_\mathrm{max}$.
\end{proof}

\begin{lemma}
For Hermitian and positive semi-definite matrix $A$,
\begin{equation}
\Tr(A^2)\le(\Tr A)^2.
\end{equation}
The equality holds if and only if only one of the eigenvalues of $A$ is nonvanishing and it is not degenerate.
\end{lemma}
\begin{proof}
The eigenvalues $\lambda_n$ of $A$ are positive semi-definite, $\lambda_n\ge0$.
Then,
\begin{equation}
\Tr(A^2)
=\sum_n\lambda_n^2
\le
\left(\sum_n\lambda_n\right)^2
=(\Tr A)^2.
\end{equation}
The equality holds if and only if $\lambda_m\lambda_n=0$ for all pairs with $m\neq n$, namely, only one of the eigenvalues $\lambda_n$ is nonvanishing and it is not degenerate.
\end{proof}

\section{Proof of the Optimality of the Measurement in Fig.~\ref{fig:MeasCircuit}}
\label{app:OptMeas}
Here, we show that the FI by the optimal input $\ket{\psi_\mathrm{opt}}$ in~(\ref{eqn:OptPure}) and the POVM $\{\hat{\Pi}_x\}$ in~(\ref{eqn:HomodynePOVM}) (the circuit in Fig.~\ref{fig:MeasCircuit}) coincides with the upper bound of the QFI in~(\ref{eqn:OptQFISeq}).
To see this, observe that the probability of measuring the value $x$ by this measurement 
in the output state of the circuit in Fig.~\ref{fig:MeasCircuit} is given by
\begin{eqnarray}
p(x|\varphi)
=
\bra{0}
\hat{S}_1^\dag(r_0)
\hat{V}_{\varphi'}^\dag
\hat{U}_\varphi^\dag
\hat{U}_{\varphi'}
\hat{V}_{\varphi'}
e^{i\theta\hat{a}_1^\dag\hat{a}_1}
\ket{x}
\bras{x}{1}
e^{-i\theta\hat{a}_1^\dag\hat{a}_1}
\hat{V}_{\varphi'}^\dag
\hat{U}_{\varphi'}^\dag
\hat{U}_\varphi
\hat{V}_{\varphi'}
\hat{S}_1(r_0)
\ket{0},
\nonumber\\
\end{eqnarray}
where the parameter $\varphi'$ used in the input state and in the measurement will be set $\varphi'=\varphi$ later.
It is the marginal of the Wigner function of the output state along the quadrature $\hat{x}_1=(\hat{a}_1+\hat{a}_1^\dag)/\sqrt{2}$.
Its characteristic function $\chi(\xi|\varphi)$ is computed to be
\begin{eqnarray}
\chi(\xi|\varphi)
&=&
\int_{-\infty}^\infty
\rmd x\,
p(x|\varphi)e^{-i\xi x}
\nonumber\\
&=&
\bra{0}
\hat{S}_1^\dag(r_0)
\hat{V}_{\varphi'}^\dag
\hat{U}_\varphi^\dag
\hat{U}_{\varphi'}
\hat{V}_{\varphi'}
e^{i\theta\hat{a}_1^\dag\hat{a}_1}
e^{-i\xi\hat{x}_1}
e^{-i\theta\hat{a}_1^\dag\hat{a}_1}
\hat{V}_{\varphi'}^\dag
\hat{U}_{\varphi'}^\dag
\hat{U}_\varphi
\hat{V}_{\varphi'}
\hat{S}_1(r_0)
\ket{0}
\vphantom{\int_{-\infty}^\infty}
\nonumber\\
&=&
e^{
-\frac{1}{2}
(\Delta x)_\theta^2
\xi^2
},
\end{eqnarray}
where 
\begin{eqnarray}
(\Delta x)_\theta^2
&=
\frac{1}{2}
\,\Bigl(
1
&
+
|
(
V_{\varphi'}^\dag
U_{\varphi'}^\dag
U_\varphi
V_{\varphi'}
)_{11}
|^2
(\cosh2r_0-1)
\nonumber\\
&&
{}
+
\Re
[
e^{-2i\theta}
(
V_{\varphi'}^\dag
U_{\varphi'}^\dag
U_\varphi
V_{\varphi'}
)_{11}^2
]
\sinh2r_0
\Bigr),
\end{eqnarray}
with $(
V_{\varphi'}^\dag
U_{\varphi'}^\dag
U_\varphi
V_{\varphi'}
)_{11}
$ being the (1,1) element of the matrix $
V_{\varphi'}^\dag
U_{\varphi'}^\dag
U_\varphi
V_{\varphi'}
$.
Its Fourier transform yields
\begin{equation}
p(x|\varphi)
=\frac{1}{\sqrt{2\pi(\Delta x)_\theta^2}}
\exp\!\left(-\frac{x^2}{2(\Delta x)_\theta^2}\right).
\end{equation}
Then, using (\ref{eqn:HamiltonianApp}) and (\ref{eqn:Hamiltonian11}), the associated FI defined by~(\ref{eqn:FI}) becomes
\begin{eqnarray}
F(\varphi|\mathcal{P},\ket{\psi_\mathrm{opt}})
&=&\int_{-\infty}^\infty \rmd x\,p(x|\varphi)\left(
\frac{\partial}{\partial\varphi}
\ln p(x|\varphi)
\right)^2
\nonumber\\
&=&
\frac{1}{2}
\left(
\frac{\partial}{\partial\varphi}\ln(\Delta x)_\theta^2
\right)^2
\nonumber\\
&=&2\varepsilon_1^2
\sinh^22r_0
\frac{4\sin^2\theta\cos^2\theta}{(e^{2r_0}\cos^2\theta+e^{-2r_0}\sin^2\theta)^2}
\end{eqnarray}
at $\varphi'=\varphi$, which can be maximized by setting $\theta=\pm\tan^{-1}e^{2r_0}$ to get
\begin{equation}
F(\varphi|\mathcal{P},\ket{\psi_\mathrm{opt}})
=
2
\varepsilon_1^2
\sinh^22r_0
=
8
\varepsilon_1^2
\overline{N}(\overline{N}+1).
\end{equation}
This coincides with the upper bound of the QFI in~(\ref{eqn:OptQFISeq}), and proves the optimality of the circuit in Fig.~\ref{fig:MeasCircuit}.

\section*{References}

\begin{thebibliography}{100}

\bibitem{ref:MetrologyScience}
Giovannetti~V, Lloyd~S and Maccone~L 2004 Quantum-Enhanced Measurements:
  Beating the Standard Quantum Limit \textit{Science} \textbf{306} 1330

\bibitem{ref:QuantumMetrologyVittorio}
Giovannetti~V, Lloyd~S and Maccone~L 2006 Quantum Metrology \textit{Phys. Rev.
  Lett.} \textbf{96} 010401

\bibitem{doi:10.1080/00107510802091298}
Dowling~J~P 2008 Quantum Optical Metrology---The Lowdown on High-{N00N} States
  \textit{Contemp. Phys.} \textbf{49} 125

\bibitem{ref:MetrologyNaturePhoto}
Giovannetti~V, Lloyd~S and Maccone~L 2011 Advances in Quantum Metrology
  \textit{Nat. Photon.} \textbf{5} 222

\bibitem{DemkowiczDobrzanski2015345}
Demkowicz-Dobrza\ifmmode~\acute{n}\else \'{n}\fi{}ski~R, Jarzyna~M and
  Ko\l{}ody\'{n}ski~J 2015 Quantum Limits in Optical Interferometry
  \textit{Progress in Optics} ed~E~Wolf (Amsterdam: Elsevier) vol~60 chap~4
  pp~345--435

\bibitem{JLightwaveTechno.33.2359}
Dowling~J~P and Seshadreesan~K~P 2015 Quantum Optical Technologies for
  Metrology, Sensing, and Imaging \textit{J. Lightwave Techno.} \textbf{33}
  2359

\bibitem{PhysRevD.23.1693}
Caves~C~M 1981 Quantum-Mechanical Noise in an Interferometer \textit{Phys. Rev.
  D} \textbf{23} 1693

\bibitem{PhysRevD.30.2548}
Bondurant~R~S and Shapiro~J~H 1984 Squeezed States in Phase-Sensing
  Interferometers \textit{Phys. Rev. D} \textbf{30} 2548

\bibitem{PhysRevA.33.4033}
Yurke~B, McCall~S~L and Klauder~J~R 1986 {SU(2)} and {SU(1,1)} Interferometers
  \textit{Phys. Rev. A} \textbf{33} 4033

\bibitem{PhysRevLett.71.1355}
Holland~M~J and Burnett~K 1993 Interferometric Detection of Optical Phase
  Shifts at the {Heisenberg} Limit \textit{Phys. Rev. Lett.} \textbf{71} 1355

\bibitem{PhysRevLett.75.2944}
Sanders~B~C and Milburn~G~J 1995 Optimal Quantum Measurements for Phase
  Estimation \textit{Phys. Rev. Lett.} \textbf{75} 2944

\bibitem{PhysRevLett.85.5098}
Berry~D~W and Wiseman~H~M 2000 Optimal States and Almost Optimal Adaptive
  Measurements for Quantum Interferometry \textit{Phys. Rev. Lett.} \textbf{85}
  5098

\bibitem{PhysRevA.73.011801}
Pezz\'e~L and Smerzi~A 2006 Phase Sensitivity of a {Mach-Zehnder}
  Interferometer \textit{Phys. Rev. A} \textbf{73} 011801(R)

\bibitem{PhysRevA.73.033821}
Monras~A 2006 Optimal Phase Measurements with Pure {Gaussian} States
  \textit{Phys. Rev. A} \textbf{73} 033821

\bibitem{PhysRevA.76.013804}
Uys~H and Meystre~P 2007 Quantum States for {Heisenberg}-Limited Interferometry
  \textit{Phys. Rev. A} \textbf{76} 013804

\bibitem{PhysRevLett.100.073601}
Pezz\'e~L and Smerzi~A 2008 {Mach-Zehnder} Interferometry at the {Heisenberg}
  Limit with Coherent and Squeezed-Vacuum Light \textit{Phys. Rev. Lett.}
  \textbf{100} 073601

\bibitem{PhysRevLett.101.253601}
Tan~S~H, Erkmen~B~I, Giovannetti~V, Guha~S, Lloyd~S, Maccone~L, Pirandola~S and
  Shapiro~J~H 2008 Quantum Illumination with {Gaussian} States \textit{Phys.
  Rev. Lett.} \textbf{101} 253601

\bibitem{PhysRevLett.102.040403}
Dorner~U, Demkowicz-Dobrza\'nski~R, Smith~B~J, Lundeen~J~S, Wasilewski~W,
  Banaszek~K and Walmsley~I~A 2009 Optimal Quantum Phase Estimation
  \textit{Phys. Rev. Lett.} \textbf{102} 040403

\bibitem{PhysRevA.79.033834}
Aspachs~M, Calsamiglia~J, {Mu\~noz-Tapia}~R and Bagan~E 2009 Phase Estimation
  for Thermal {Gaussian} States \textit{Phys. Rev. A} \textbf{79} 033834

\bibitem{PhysRevLett.102.253601}
Tsang~M 2009 Quantum Imaging beyond the Diffraction Limit by Optical Centroid
  Measurements \textit{Phys. Rev. Lett.} \textbf{102} 253601

\bibitem{PhysRevLett.104.103602}
Anisimov~P~M, Raterman~G~M, Chiruvelli~A, Plick~W~N, Huver~S~D, Lee~H and
  Dowling~J~P 2010 Quantum Metrology with Two-Mode Squeezed Vacuum: Parity
  Detection Beats the {Heisenberg} Limit \textit{Phys. Rev. Lett.} \textbf{104}
  103602

\bibitem{PhysRevLett.105.120501}
Hyllus~P, Pezz\'e~L and Smerzi~A 2010 Entanglement and Sensitivity in Precision
  Measurements with States of a Fluctuating Number of Particles \textit{Phys.
  Rev. Lett.} \textbf{105} 120501

\bibitem{Escher:2011aa}
Escher~B~M, de~Matos~Filho~R~L and Davidovich~L 2011 General Framework for
  Estimating the Ultimate Precision Limit in Noisy Quantum-Enhanced Metrology
  \textit{Nat. Phys.} \textbf{7} 406

\bibitem{PhysRevLett.107.083601}
Joo~J, Munro~W~J and Spiller~T~P 2011 Quantum Metrology with Entangled Coherent
  States \textit{Phys. Rev. Lett.} \textbf{107} 083601

\bibitem{PhysRevA.85.010101}
Pinel~O, Fade~J, Braun~D, Jian~P, Treps~N and Fabre~C 2012 Ultimate Sensitivity
  of Precision Measurements with Intense {Gaussian} Quantum Light: A Multimodal
  Approach \textit{Phys. Rev. A} \textbf{85} 010101(R)

\bibitem{PhysRevA.85.011801}
Jarzyna~M and Demkowicz-Dobrza\ifmmode~\acute{n}\else \'{n}\fi{}ski~R 2012
  Quantum Interferometry with and without an External Phase Reference
  \textit{Phys. Rev. A} \textbf{85} 011801(R)

\bibitem{ref:RivasLuis-NJP2012}
Rivas~{\'A} and Luis~A 2012 Sub-Heisenberg Estimation of Non-Random Phase
  Shifts \textit{New J. Phys.} \textbf{14} 093052

\bibitem{PhysRevA.87.012107}
Genoni~M~G, Paris~M~G~A, Adesso~G, Nha~H, Knight~P~L and Kim~M~S 2013 Optimal
  Estimation of Joint Parameters in Phase Space \textit{Phys. Rev. A}
  \textbf{87} 012107

\bibitem{arXiv:1303.3682}
Monras~A 2013 Phase Space Formalism for Quantum Estimation of {Gaussian} States
  \textit{arXiv:1303.3682 [quant-ph]}

\bibitem{PhysRevLett.110.163604}
Pezz\'e~L and Smerzi~A 2013 Ultrasensitive Two-Mode Interferometry with
  Single-Mode Number Squeezing \textit{Phys. Rev. Lett.} \textbf{110} 163604

\bibitem{PhysRevLett.110.213601}
Ruo~Berchera~I, Degiovanni~I~P, Olivares~S and Genovese~M 2013 Quantum Light in
  Coupled Interferometers for Quantum Gravity Tests \textit{Phys. Rev. Lett.}
  \textbf{110} 213601

\bibitem{PhysRevA.88.013838}
Zhang~X~X, Yang~Y~X and Wang~X~B 2013 Lossy Quantum-Optical Metrology with
  Squeezed States \textit{Phys. Rev. A} \textbf{88} 013838

\bibitem{PhysRevLett.111.173601}
Lang~M~D and Caves~C~M 2013 Optimal Quantum-Enhanced Interferometry Using a
  Laser Power Source \textit{Phys. Rev. Lett.} \textbf{111} 173601

\bibitem{PhysRevA.88.040102}
Pinel~O, Jian~P, Treps~N, Fabre~C and Braun~D 2013 Quantum Parameter Estimation
  Using General Single-Mode {Gaussian} States \textit{Phys. Rev. A} \textbf{88}
  040102(R)

\bibitem{PhysRevA.88.063820}
Sahota~J and James~D~F~V 2013 Quantum-Enhanced Phase Estimation with an
  Amplified {Bell} State \textit{Phys. Rev. A} \textbf{88} 063820

\bibitem{PhysRevA.89.032128}
Jiang~Z 2014 Quantum {Fisher} Information for States in Exponential Form
  \textit{Phys. Rev. A} \textbf{89} 032128

\bibitem{PhysRevA.89.053822}
Tan~Q~S, Liao~J~Q, Wang~X and Nori~F 2014 Enhanced Interferometry Using
  Squeezed Thermal States and Even or Odd States \textit{Phys. Rev. A}
  \textbf{89} 053822

\bibitem{PhysRevA.90.025802}
Lang~M~D and Caves~C~M 2014 Optimal Quantum-Enhanced Interferometry
  \textit{Phys. Rev. A} \textbf{90} 025802

\bibitem{PhysRevA.90.033846}
Knott~P~A, Proctor~T~J, Nemoto~K, Dunningham~J~A and Munro~W~J 2014 Effect of
  Multimode Entanglement on Lossy Optical Quantum Metrology \textit{Phys. Rev.
  A} \textbf{90} 033846

\bibitem{PhysRevA.91.013808}
Sahota~J and Quesada~N 2015 Quantum Correlations in Optical Metrology:
  Heisenberg-Limited Phase Estimation without Mode Entanglement \textit{Phys.
  Rev. A} \textbf{91} 013808

\bibitem{PhysRevA.91.032103}
Pezz\`e~L, Hyllus~P and Smerzi~A 2015 Phase-Sensitivity Bounds for Two-Mode
  Interferometers \textit{Phys. Rev. A} \textbf{91} 032103

\bibitem{PhysRevLett.114.170802}
Motes~K~R, Olson~J~P, Rabeaux~E~J, Dowling~J~P, Olson~S~J and Rohde~P~P 2015
  Linear Optical Quantum Metrology with Single Photons: Exploiting
  Spontaneously Generated Entanglement to Beat the Shot-Noise Limit
  \textit{Phys. Rev. Lett.} \textbf{114} 170802

\bibitem{Sparaciari:15}
Sparaciari~C, Olivares~S and Paris~M~G~A 2015 Bounds to Precision for Quantum
  Interferometry with {Gaussian} States and Operations \textit{J. Opt. Soc. Am.
  B} \textbf{32} 1354

\bibitem{PhysRevA.92.042331}
Rigovacca~L, Farace~A, De~Pasquale~A and Giovannetti~V 2015 Gaussian
  Discriminating Strength \textit{Phys. Rev. A} \textbf{92} 042331

\bibitem{1367-2630-17-7-073016}
\v{S}afr\'{a}nek~D, Lee~A~R and Fuentes~I 2015 Quantum Parameter Estimation
  Using Multi-Mode {Gaussian} States \textit{New J. Phys.} \textbf{17} 073016

\bibitem{PhysRevA.92.022106}
Friis~N, Skotiniotis~M, Fuentes~I and D\"ur~W 2015 Heisenberg Scaling in
  {Gaussian} Quantum Metrology \textit{Phys. Rev. A} \textbf{92} 022106

\bibitem{ref:EstId}
De~Pasquale~A, Facchi~P, Florio~G, Giovannetti~V, Matsuoka~K and Yuasa~K 2015
  Two-Mode Bosonic Quantum Metrology with Number Fluctuations \textit{Phys.
  Rev. A} \textbf{92} 042115

\bibitem{PhysRevA.93.013809}
Gao~Y and Wang~R~m 2016 Variational Limits for Phase Precision in Linear
  Quantum Optical Metrology \textit{Phys. Rev. A} \textbf{93} 013809

\bibitem{PhysRevA.93.023810}
Sparaciari~C, Olivares~S and Paris~M~G~A 2016 Gaussian-State Interferometry
  with Passive and Active Elements \textit{Phys. Rev. A} \textbf{93} 023810

\bibitem{PhysRevA.93.033859}
Knott~P~A, Proctor~T~J, Hayes~A~J, Cooling~J~P and Dunningham~J~A 2016
  Practical Quantum Metrology with Large Precision Gains in the
  Low-Photon-Number Regime \textit{Phys. Rev. A} \textbf{93} 033859

\bibitem{PhysRevA.94.023834}
Gao~Y 2016 Quantum Optical Metrology in the Lossy {SU(2)} and {SU(1,1)}
  Interferometers \textit{Phys. Rev. A} \textbf{94} 023834

\bibitem{PhysRevX.6.031033}
Tsang~M, Nair~R and Lu~X~M 2016 Quantum Theory of Superresolution for Two
  Incoherent Optical Point Sources \textit{Phys. Rev. X} \textbf{6} 031033

\bibitem{PhysRevA.94.033817}
Sahota~J, Quesada~N and James~D~F~V 2016 Physical Resources for Optical Phase
  Estimation \textit{Phys. Rev. A} \textbf{94} 033817

\bibitem{PhysRevA.94.042327}
Volkoff~T~J 2016 Optimal and Near-Optimal Probe States for Quantum Metrology of
  Number-Conserving Two-Mode Bosonic {Hamiltonians} \textit{Phys. Rev. A}
  \textbf{94} 042327

\bibitem{PhysRevA.94.042342}
Gagatsos~C~N, Branford~D and Datta~A 2016 Gaussian Systems for Quantum-Enhanced
  Multiple Phase Estimation \textit{Phys. Rev. A} \textbf{94} 042342

\bibitem{PhysRevLett.117.190801}
Nair~R and Tsang~M 2016 Far-Field Superresolution of Thermal Electromagnetic
  Sources at the Quantum Limit \textit{Phys. Rev. Lett.} \textbf{117} 190801

\bibitem{PhysRevLett.117.190802}
Lupo~C and Pirandola~S 2016 Ultimate Precision Bound of Quantum and
  Subwavelength Imaging \textit{Phys. Rev. Lett.} \textbf{117} 190802

\bibitem{PhysRevX.6.041044}
Oszmaniec~M, Augusiak~R, Gogolin~C, Ko\l{}ody\ifmmode~\acute{n}\else
  \'{n}\fi{}ski~J, Ac\'{\i}n~A and Lewenstein~M 2016 Random Bosonic States for
  Robust Quantum Metrology \textit{Phys. Rev. X} \textbf{6} 041044

\bibitem{PhysRevA.94.062313}
\ifmmode~\check{S}\else \v{S}\fi{}afr\'anek~D and Fuentes~I 2016 Optimal Probe
  States for the Estimation of {Gaussian} Unitary Channels \textit{Phys. Rev.
  A} \textbf{94} 062313

\bibitem{PhysRevA.95.012109}
Jarzyna~M and Zwierz~M 2017 Parameter Estimation in the Presence of the Most
  General {Gaussian} Dissipative Reservoir \textit{Phys. Rev. A} \textbf{95}
  012109

\bibitem{ref:MultiParGaussianMetrology}
Nichols~R, Liuzzo-Scorpo~P, Knott~P~A and Adesso~G 2018 Multiparameter Gaussian
  Quantum Metrology \textit{Phys. Rev. A} \textbf{98} 012114

\bibitem{arXiv:1701.05152}
Braun~D, Adesso~G, Benatti~F, Floreanini~R, Marzolino~U, Mitchell~M~W and
  Pirandola~S 2018 Quantum-Enhanced Measurements without Entanglement
  \textit{Rev. Mod. Phys.} \textbf{90} 035006

\bibitem{arXiv:1602.05958}
Spedalieri~G, Lupo~C, Braunstein~S~L and Pirandola~S 2019 Thermal Quantum
  Metrology in Memoryless and Correlated Environments \textit{Quantum Sci.
  Technol.} \textbf{4} 015008

\bibitem{arXiv:1801.00299}
\v{S}afr\'{a}nek~D 2019 Estimation of Gaussian Quantum States \textit{J. Phys.
  A: Math. Theor.} \textbf{52} 035304

\bibitem{Bouwmeester:2004aa}
Bouwmeester~D 2004 Quantum Physics: High {NOON} for Photons \textit{Nature
  (London)} \textbf{429} 139

\bibitem{Nagata726}
Nagata~T, Okamoto~R, O{\textquoteright}Brien~J~L, Sasaki~K and Takeuchi~S 2007
  Beating the Standard Quantum Limit with Four-Entangled Photons
  \textit{Science} \textbf{316} 726

\bibitem{Higgins:2007aa}
Higgins~B~L, Berry~D~W, Bartlett~S~D, Wiseman~H~M and Pryde~G~J 2007
  Entanglement-Free {Heisenberg}-limited Phase Estimation \textit{Nature
  (London)} \textbf{450} 393

\bibitem{OBrien:2009aa}
O'Brien~J~L, Furusawa~A and Vuckovic~J 2009 Photonic Quantum Technologies
  \textit{Nat. Photon.} \textbf{3} 687

\bibitem{BridaG.:2010aa}
Brida~G, Genovese~M and Ruo~Berchera~I 2010 Experimental Realization of
  Sub-Shot-Noise Quantum Imaging \textit{Nat. Photon.} \textbf{4} 227

\bibitem{Afek879}
Afek~I, Ambar~O and Silberberg~Y 2010 High-{NOON} States by Mixing Quantum and
  Classical Light \textit{Science} \textbf{328} 879

\bibitem{KacprowiczM.:2010aa}
Kacprowicz~M, Demkowicz-Dobrza\'nski~R, Wasilewski~W, Banaszek~K and
  Walmsley~I~A 2010 Experimental Quantum-Enhanced Estimation of a Lossy Phase
  Shift \textit{Nat. Photon.} \textbf{4} 357

\bibitem{Xian:2011aa}
Xiang~G~Y, Higgins~B~L, Berry~D~W, Wiseman~H~M and Pryde~G~J 2011
  Entanglement-Enhanced Measurement of a Completely Unknown Optical Phase
  \textit{Nat. Photon.} \textbf{5} 43

\bibitem{PhysRevLett.107.080504}
Krischek~R, Schwemmer~C, Wieczorek~W, Weinfurter~H, Hyllus~P, Pezz\'e~L and
  Smerzi~A 2011 Useful Multiparticle Entanglement and Sub-Shot-Noise
  Sensitivity in Experimental Phase Estimation \textit{Phys. Rev. Lett.}
  \textbf{107} 080504

\bibitem{PhysRevA.85.043817}
Genoni~M~G, Olivares~S, Brivio~D, Cialdi~S, Cipriani~D, Santamato~A, Vezzoli~S
  and Paris~M~G~A 2012 Optical Interferometry in the Presence of Large Phase
  Diffusion \textit{Phys. Rev. A} \textbf{85} 043817

\bibitem{doi:10.1063/1.4724105}
Crespi~A, Lobino~M, Matthews~J~C~F, Politi~A, Neal~C~R, Ramponi~R, Osellame~R
  and O'Brien~J~L 2012 Measuring Protein Concentration with Entangled Photons
  \textit{Appl. Phys. Lett.} \textbf{100} 233704

\bibitem{Wolfgramm:2013aa}
Wolfgramm~F, Vitelli~C, Beduini~F~A, Godbout~N and Mitchell~M~W 2013
  Entanglement-Enhanced Probing of a Delicate Material System \textit{Nat.
  Photon.} \textbf{7} 28

\bibitem{Taylor:2013aa}
Taylor~M~A, Janousek~J, Daria~V, Knittel~J, Hage~B, Bachor~H~A and Bowen~W~P
  2013 Biological Measurement beyond the Quantum Limit \textit{Nat. Photon.}
  \textbf{7} 229

\bibitem{Ono:2013aa}
Ono~T, Okamoto~R and Takeuchi~S 2013 An Entanglement-Enhanced Microscope
  \textit{Nat. Commun.} \textbf{4} 2426

\bibitem{Vidrighin:2014aa}
Vidrighin~M~D, Donati~G, Genoni~M~G, Jin~X~M, Kolthammer~W~S, Kim~M~S, Datta~A,
  Barbieri~M and Walmsley~I~A 2014 Joint Estimation of Phase and Phase
  Diffusion for Quantum Metrology \textit{Nat. Commun.} \textbf{5} 3532

\bibitem{PhysRevLett.112.103604}
Israel~Y, Rosen~S and Silberberg~Y 2014 Supersensitive Polarization Microscopy
  Using {NOON} States of Light \textit{Phys. Rev. Lett.} \textbf{112} 103604

\bibitem{PhysRevLett.112.223602}
Rozema~L~A, Bateman~J~D, Mahler~D~H, Okamoto~R, Feizpour~A, Hayat~A and
  Steinberg~A~M 2014 Scalable Spatial Superresolution Using Entangled Photons
  \textit{Phys. Rev. Lett.} \textbf{112} 223602

\bibitem{ref:ContVarQI}
Braunstein~S~L and van Loock~P 2005 Quantum Information with Continuous
  Variables \textit{Rev. Mod. Phys.} \textbf{77} 513

\bibitem{ref:GaussianQI}
Weedbrook~C, Pirandola~S, {Garc\'{\i}a-Patr\'on}~R, Cerf~N~J, Ralph~T~C,
  Shapiro~J~H and Lloyd~S 2012 Gaussian Quantum Information \textit{Rev. Mod.
  Phys.} \textbf{84} 621

\bibitem{ref:ContVarQI-Adesso}
Adesso~G, Ragy~S and Lee~A~R 2014 Continuous Variable Quantum Information:
  Gaussian States and Beyond \textit{Open Sys. Inf. Dyn.} \textbf{21} 1440001

\bibitem{ref:Helstrom}
Helstrom~C~W 1976 \textit{Quantum Detection and Estimation Theory} (New York:
  Academic Press)

\bibitem{ref:BraunsteinCave1994}
Braunstein~S~L and Caves~C~M 1994 Statistical Distance and the Geometry of
  Quantum States \textit{Phys. Rev. Lett.} \textbf{72} 3439

\bibitem{ref:BraunsteinCave1996AnnPhys}
Braunstein~S~L, Caves~C~M and Milburn~G~J 1996 Generalized Uncertainty
  Relations: Theory, Examples, and {Lorentz} Invariance \textit{Ann. Phys.
  (N.Y.)} \textbf{247} 135

\bibitem{ref:HayashiAsymptoticTheory}
Hayashi~M 2005 \textit{Asymptotic Theory of Quantum Statistical Inference:
  Selected Papers} (Singapore: World Scientific)

\bibitem{ref:Paris-IJQI}
Paris~M~G~A 2009 Quantum Estimation for Quantum Technology \textit{Int. J.
  Quant. Inf.} \textbf{7} 125

\bibitem{ref:HolevoSNS}
Holevo~A~S 2011 \textit{Probabilistic and Statistical Aspects of Quantum
  Theory} (Pisa: Edizioni della Normale)

\bibitem{ref:NielsenChuang}
Nielsen~M~A and Chuang~I~L 2000 \textit{Quantum Computation and Quantum
  Information} (Cambridge: Cambridge University Press)

\bibitem{ref:HayashiIshizukaKawachiKimuraOgawa}
Hayashi~M, Ishizaka~S, Kawachi~A, Kimura~G and Ogawa~T 2015
  \textit{Introduction to Quantum Information Science} (Berlin: Springer)

\bibitem{ref:GaussFidelity}
Banchi~L, Braunstein~S~L and Pirandola~S 2015 Quantum Fidelity for Arbitrary
  {Gaussian} States \textit{Phys. Rev. Lett.} \textbf{115} 260501

\bibitem{ref:QFI-convexity-Fujiwara2001}
Fujiwara~A 2001 Quantum Channel Identification Problem \textit{Phys. Rev. A}
  \textbf{63} 042304

\bibitem{ref:EstCtrlInversion}
Yuan~H and Fung~C~H~F 2015 Optimal Feedback Scheme and Universal Time Scaling
  for {Hamiltonian} Parameter Estimation \textit{Phys. Rev. Lett.} \textbf{115}
  110401

\bibitem{PhysRevA.70.032310}
Ballester~M~A 2004 Entanglement is Not Very Useful for Estimating Multiple
  Phases \textit{Phys. Rev. A} \textbf{70} 032310

\bibitem{PhysRevLett.120.080501}
Proctor~T~J, Knott~P~A and Dunningham~J~A 2018 Multiparameter Estimation in
  Networked Quantum Sensors \textit{Phys. Rev. Lett.} \textbf{120} 080501

\bibitem{ref:MatsuokaMS}
Matsuoka~K {Master's Thesis} (in {Japanese}) Waseda University Tokyo 2015

\end{thebibliography}

\end{document}